[1994/12/01]% LaTeX date must December 1994 or later
\documentclass[10pt, reqno]{amsart}
\usepackage{a4wide}
\usepackage[margin=1.9cm]{geometry}
\usepackage[english, activeacute]{babel}
\usepackage{amsmath,amsthm,amsxtra}
\usepackage{epsfig}
\usepackage{amssymb}
\usepackage{latexsym}
\usepackage{amsfonts}
\usepackage{hyperref}
\usepackage{multicol}
\title{Nonlinear stability of higher order mKdV breathers}
\author{Miguel A. Alejo}
\thanks{M. A. was partially funded by CNPq grant no. 305205/2016-1.}
\author{Eleomar Cardoso}
\address{Departamento de Matem\'atica, Universidade Federal de Santa Catarina, Brasil}
\email{miguel.alejo@ufsc.br}

\email{eleomar.junior@ufsc.br}
\date{\today}
\thanks{}

\chardef\bslash=`\\ % p. 424, TeXbook
\newtheorem{thm}{Theorem}[section]
\newtheorem{cor}[thm]{Corollary}
\newtheorem{lem}[thm]{Lemma}

\newtheorem{defn}[thm]{Definition}

\theoremstyle{remark}
\newtheorem{rem}{Remark}[section]

\numberwithin{equation}{section}

\newcommand{\R}{\mathbb{R}}
\newcommand{\N}{\mathbb{N}}

\newcommand{\al}{\alpha}
\newcommand{\bt}{\beta}
\newcommand{\ga}{\gamma}

\newcommand{\spawn}{\operatorname{span}}
\newcommand{\sech}{\operatorname{sech}}

\newcommand{\be}{\begin{equation}}
\newcommand{\ee}{\end{equation}}
\newcommand{\bp}{\begin{proof}}
\newcommand{\ep}{\end{proof}}
\newcommand{\bel}{\begin{equation}\label}
\newcommand{\eeq}{\end{equation}}
\newcommand{\bea}{\begin{eqnarray}}
\newcommand{\eea}{\end{eqnarray}}
\newcommand{\bee}{\begin{eqnarray*}}
\newcommand{\eee}{\end{eqnarray*}}
\newcommand{\ben}{\begin{enumerate}}
\newcommand{\een}{\end{enumerate}}
\newcommand{\nonu}{\nonumber}

\newcommand{\ms}{\medskip}

\newcommand{\eval}[2][\right]{\relax
  \ifx#1\right\relax \left.\fi#2#1\rvert}

\begin{document}
\begin{abstract}
We are interested in stability results for breather solutions of the 5th, 7th and 9th order mKdV equations.
We  show that these higher order mKdV breathers are  stable in $H^2(\R)$, in the same way as \emph{classical} mKdV breathers. 
We also show that breather solutions of the 5th, 7th and 9th order mKdV equations satisfy  the same stationary fourth order
nonlinear elliptic equation as the mKdV breather, independently of the order, 5th, 7th or 9th, considered.

% overcoming the problem of the loss of the scaling property of this equation.
% We measure the regularity of the Cauchy problem in the classical Sobolev spaces $H^s$,
% and show the critical Sobolev index under which the local well-posedness  of the problem is not present,
% in the sense that the dependence of solutions upon initial data fails to be continuous.
\end{abstract}
\maketitle \markboth{Nonlinear stability of higher order mKdV breathers}{Miguel A. Alejo and Eleomar Cardoso}
\renewcommand{\sectionmark}[1]{}

%%%%%%%%%%%%%%%%%%%%%%%%%%%%%%%%%%%%%%%%%%%%%%%%%%%%%%%%%%%%%%%%%%%%%%%%%%%%%%%%%%%%%%%%%%%%%%%%%%%%%%%%%%%%%%%%%%%%%%%%%%%%%%%%%%%%%%%%%%%%%%%%%%%%%%%%%%%%%%%%%%%%%%%%%%%%%%%%%%%%%%%%%%%%%%%%%%%%%%%%%%%%%%%%%%%%%%%%%%%%%%%%%%%%%%%%%%%%%%%%%%%%%%%%%%%%%%%%%%%%%%%%%%%%%%%%%%%%%%%%%%%%%%%%%%%%%%%%%%%%%%%
\section{Introduction}
In this note  we extend  previous results on the stability of breather solutions of the \emph{focusing} modified Korteweg-de Vries (mKdV) equation (see \cite{AM}),
%on stability properties  for breather solutions of the modified Korteweg-de Vries (mKdV) equation,
\be\label{mkdv}
%u_{t} +(u_{4x} +5uu_x^2 +5u^2u_{2x} + \frac{3}{2}u^5)_x=0,
%u_{t} +(u_{4x} +10uu_x^2 +10u^2u_{xx} + 6u^5)_x=0,
u_{t} +(u_{xx}+ 2u^3)_x=0,\quad u(t,x)\in\R,\\
\ee
\noindent
to new breather solutions of  higher order \emph{focusing} versions of \eqref{mkdv}. Namely, we are going to focus on\newline

\medskip
\noindent
the \emph{focusing} 5th-order mKdV equation

\be\begin{aligned}\label{5mkdv}
& u_{t} +(u_{4x} + f_5(u))_x=0,\\
&\\
&f_5(u):=~10uu_x^2 +10u^2u_{xx} + 6u^5,\\
%u_{t} +(u_{6x} +14u^2u_{4x} +56uu_xu_{3x} + 42uu_{2x}^2 + 70u_x^2u_{2x} + 70u^4u_{2x} + 140u^3u_x^2 + 20u^7)_x=0,\nonu\\
\end{aligned}\ee
\medskip
\noindent
 the \emph{focusing} 7th-order mKdV equation
 
\be\begin{aligned}\label{7mkdv}
&u_{t} +(u_{6x} + f_7(u))_x=0,\\
&\\
&f_7(u):=~14u^2u_{4x} +56uu_xu_{3x} + 42uu_{xx}^2 + 70u_x^2u_{xx} + 70u^4u_{xx} + 140u^3u_x^2 + 20u^7,
\end{aligned}\ee

\medskip
\noindent
and the \emph{focusing} 9th-order mKdV 

\be\label{9mkdv}\begin{aligned}
&u_{t} + (u_{8x} + f_9(u))_x=0,\\
&\\
&f_9(u):=~18u^2u_{6x} + 108uu_xu_{5x} +228uu_{2x}u_{4x} + 210(u_x)^2u_{4x} + 126u^4u_{4x} + 138u(u_{3x})^2 \\
& + 756u_xu_{2x}u_{3x} + 1008u^3u_xu_{3x} + 182(u_{2x})^3 + 756u^3(u_{2x})^2 + 3108u^2(u_x)^2u_{2x}  \\
& + 420u^6u_{2x} + 798u(u_x)^4 + 1260u^5(u_x)^2 + 70u^9,\\
&\\
\end{aligned}\ee

%\medskip
\noindent 
and which we will denote them as 5th, 7th and 9th-mKdV equations hereafter. All these higher order mKdV equations are members of an infinite family of equations, the so call
\emph{focusing} mKdV hierarchy of equations, as it was shown by Alejo-Cardoso \cite{AleCar} (see \cite{JFG} for a \emph{defocusing} mKdV version of this hierarchy).
Note that we are only interested in \emph{focusing} mKdV versions since these models are the only mKdV equations bearing regular (not singular) and real \emph{breather} solutions. 
Moreover, other  higher order mKdV cases, (e.g. $(2n+1)th$-mKdV,~$n\geq5$) will not be treated here, since beside increasing the number of terms in each equation of the higher 
order hierarchy (see Appendix \ref{9mkdvApp}), we have not 
at hand a global well posedness theory of them in a Sobolev space $H^s(\R)$ with $s\leq2$, as it was pointed out by Gr\"unrock \cite[p.506,~Cor.2.1]{Gru}. 
Since our stability result is stated taking into account small perturbations in $H^2(\R)$, only  higher order mKdV equations with a Cauchy problem defined in a Sobolev space $H^s(\R),~s\leq2$  are going to be considered here, namely the 5th, 7th and 9th-mKdV equations (see \cite{Lin}, \cite{Gru} for further reading).

%but this approach will be sufficient for our purposes. 
%Here $u = u(t,x)$ is a real valued function.

%Specifically, we will
%work with breather solutions of \eqref{5mkdv} and \eqref{7mkdv} and
%we will present here one result about the stability in $H^2$ in a
%similar way that we did in the mKdV case \eqref{mkdv}.

\medskip

% These higher order mKdV equations are a well-known \emph{completely integrable} set of models \cite{Ga,AC,La}, with infinitely many
% conservation laws and well-known (long-time) asymptotic behavior of its solutions obtained with the help of the inverse scattering
% transform \cite{GrSl}. As a physical model, \eqref{5mkdv} and \eqref{7mkdv} describe large-amplitude internal solitary waves,
% showing a dynamics which can look rather different from the KdV form. On the other hand, solutions of \eqref{5mkdv},
% \eqref{7mkdv} and \eqref{9mkdv} are invariant under space and time translations. Indeed, for any $t_0, x_0\in \R$, $u(t-t_0, x-x_0)$ is also a
% solution of both equations. Linares by using a contraction mapping argument, in \cite{Lin}, has showed that the Cauchy problem for the
% 5th-mKdV equation is locally well-posed at $H^2(\R)$. Kwon, \cite{Kwon}, obtained a better result: the 5th-mKdV equation is
% locally well-posed at $H^s(\R),~s\geq\frac{3}{4}$. Finally, Gr\"unrock, \cite{Gru}, deduced well-possedness results to other
% higher-order mKdV equations at Theorem 2.1. This author established that 7th-mKdV equation is locally well-posed at
% $H^s(\R),~s\geq\frac{5}{4}$. The Cauchy problem for the 5th-mKdV equation is globally well-posed at $H^s(\R),~s\geq1$ and in the case
% of the 7th-mKdV and \eqref{9mkdv} equations at $H^s(\R),~s\geq2$. See e.g. Linares \cite{Lin}, Kwon \cite{Kwon} and Gr\"unrock \cite{Gru} for further details. 
% Note moreover that we have the following inner relation between mKdV
These higher order mKdV equations are a well-known \emph{completely integrable} set of models \cite{AC,Ga,La}, with infinitely many
conservation laws. On the other hand, solutions of \eqref{5mkdv}, \eqref{7mkdv} and \eqref{9mkdv} are invariant under space and time translations. 
Indeed, for any $t_0, x_0\in \R$, $u(t-t_0, x-x_0)$ is also a solution of both equations. Even more,  $-u$ is also a solution of \eqref{5mkdv}, \eqref{7mkdv} 
and \eqref{9mkdv} for any previously given solution $u$.

\medskip

About the Cauchy problem of higher order versions of \eqref{mkdv}, Linares by using a contraction mapping argument showed  in \cite{Lin} that the initial value problem for the
5th-mKdV equation is locally well-posed at $H^2(\R)$. Kwon, \cite{Kwon}, obtained a better result: the 5th-mKdV equation is
locally well-posed at $H^s(\R),~s\geq\frac{3}{4}$. Finally, Gr\"unrock, \cite{Gru}, deduced well-possedness results to other
higher-order mKdV equations at Theorem 2.1. This same author established that 7th-mKdV equation is locally well-posed at
$H^s(\R),~s\geq\frac{5}{4}$. The Cauchy problem for the 5th-mKdV equation is globally well-posed at $H^s(\R),~s\geq1$ and in the case
of the 7th and 9th-mKdV \eqref{7mkdv}-\eqref{9mkdv} equations at $H^s(\R),~s\geq2$. See e.g. Linares \cite{Lin}, Kwon \cite{Kwon} and Gr\"unrock \cite{Gru} for further details. 
Note moreover that we have the following inner relation between mKdV
%\begin{itemize}
% \item[mKdV]
 \be\label{v2mkdv}
u_{t} = - \partial_x(u_{xx}+ 2u^3),
\ee
%\item[5th-mKdV]
and its higher order versions, namely the 5th-mKdV,
\be\begin{aligned}\label{v25mkdv}
& u_{t} = -\partial_{x}\Big(\partial_{x}^2(u_{xx}+ 2u^3) - (2uu_x^2 - 4u^2u_{xx} - 6u^5)\Big),
%u_{t} +(u_{6x} +14u^2u_{4x} +56uu_xu_{3x} + 42uu_{2x}^2 + 70u_x^2u_{2x} + 70u^4u_{2x} + 140u^3u_x^2 + 20u^7)_x=0,\nonu\\
\end{aligned}\ee
\medskip
\noindent
 the 7th-mKdV
\be\begin{aligned}\label{v27mkdv}
%u_{t} +(u_{4x} +5uu_x^2 +5u^2u_{2x} + \frac{3}{2}u^5)_x=0,
%u_{t} +(u_{4x} +10uu_x^2 +10u^2u_{xx} + 6u^5)_x=0,
u_{t} =&  -\partial_{x}\Big(\partial_{x}^4(u_{xx}+ 2u^3) - \partial_x^2(2uu_x^2 - 4u^2u_{xx} - 6u^5) \\
&- (4uu_xu_{3x} - 4u^2u_{4x} - 2uu_{xx}^2 - 40u^4u_{xx} - 20u^3u_x^2 - 20u^7)\Big).
\end{aligned}\ee
\medskip
\noindent
and the 9th-mKdV
\begin{center}
\be\begin{aligned}\label{v29mkdv}
& \qquad \qquad\qquad\qquad\qquad u_{t} =  -\partial_{x}\Big[\partial_{x}^6(u_{xx}+ 2u^3) - \partial_x^4(2uu_x^2 - 4u^2u_{xx} - 6u^5)\\ 
& \qquad \qquad\qquad - \partial_x^2(4uu_xu_{3x} - 4u^2u_{4x} - 2uu_{xx}^2 - 40u^4u_{xx} - 20u^3u_x^2 - 20u^7) +\Big(8u^2u_{6x} - 4uu_{6x}\\
&+ 26uu_xu_{5x} - 16 u_xu_{5x} + 52 uu_{xx}u_{4x} - 28u_{xx}u_{4x} + 39u_x^2u_{4x} +39u_{xx}u_{4x} + 56u^4u_{4x} + 24uu_{3x}^2 - 16u_{3x}^2\\
& + 84u_xu_{xx}u_{3x} + 168u^3u_xu_{3x} + 12u_{xx}^3 + 196u^3u_{xx}^2 + 168u^2u_x^2u_{xx} + 280u^6u_{xx} - 42uu_x^4 + 420u^5u_x^2 + 70u^9\Big)\Big].
\end{aligned}\ee
\end{center}
\noindent
In the case of the 5th, 7th and 9th-mKdV equations \eqref{5mkdv}-\eqref{7mkdv}-\eqref{9mkdv}, the profile of their soliton solutions is
completely similar to the well known {\it$\sech$}  mKdV soliton profile,
and it is explicitly given   by the formula (we denote by $v_5,v_7,v_9$ the speeds of 5th, 7th and 9th order solitons)
%\bel{fexplicita2}Q_{c,\mu} (s) :=\frac{c}{\mu^2(\frac{\mu}{3}+\sqrt{\frac{\mu^2}{9}+\frac{c}{2}}\cosh(\mu^{\frac{-3}{2}}\sqrt{c}s))}.\eeq
\be\label{fexplicita2}\begin{array}{ll} u(t,x) := Q_{c} (x-v_it)|_{i=5,7,9},\quad v_5=c^2,~v_7=c^3,~v_9=c^4\\\\
%\qquad Q_{c,\mu} (s) :=\frac{c}{ \displaystyle{ \frac{\mu}{3} +\sqrt{\frac{\mu^2}{9}+\frac{c}{2}}\cosh(\sqrt{c}s) } }.\eeq
Q_{c} (s) :=\sqrt{c} \sech(\sqrt{c} s),~c>0.\end{array}\ee
%
%Since both 5th and 7th-mKdV soliton profiles \eqref{fexplicita2} are the same profile as the mKdV soliton solution,
Moreover, it is easy to see, by substitution that both
5th, 7th and 9th-mKdV soliton solutions $Q_{c}$ \eqref{fexplicita2} satisfy the \emph{same} nonlinear stationary elliptic equation
\be\label{eqQc}\begin{array}{ll}
Q_c'' -c\, Q_c + 2Q_c^3=0, \quad Q_c>0, \quad Q_{c}\in H^1(\R).\\
%Q_{c}^{iv} -c^4\, Q_{c} + 10(Q_{c}^{'})^2Q_{c} + 10Q_{c}^2Q_{c}^{''} + 6Q_{c}^5=0, \quad Q_{c}>0, \quad Q_{c}\in H^1(\R).
\end{array}\ee
\noindent
Note that this second order ODE is precisely the one satisfied by the mKdV \emph{classical} soliton. Moreover, note that the soliton
solution \eqref{fexplicita2} of the 5th, 7th and 9th-mKdV equations also satisfy the 4th, 6th and 8th order elliptic ODEs  coming naturally from
integration in space of the 5th, 7th and 9th order mKdV equations \eqref{5mkdv}-\eqref{7mkdv}-\eqref{9mkdv} respectively. Namely, 5th, 7th and 9th higher
order mKdV solitons satisfy the following nonlinear stationary elliptic equations:

\be\label{5thODE}
 Q_{c}^{(iv)} -c^2\, Q_{c} + f_5(Q_c) =0,
%10(Q_{c}^{'})^2Q_{c} + 10Q_{c}^2Q_{c}^{''} + 6Q_{c}^5=0,
%E[u](t)  :=  \frac 12 \int_\R u_x^2(t,x)dx -\frac 14 \int_\R u^4(t,x)dx = E[u](0).
\ee
\noindent
\medskip
\noindent
\begin{align}\label{7thODE}
%Q_c'' -c\, Q_c + 2Q_c^3=0, \quad Q_c>0, , \quad Q_{c}>0, \quad Q_{c}\in H^1(\R).\\
%& Q_{c}^{iv} -c^2\, Q_{c} + 10(Q_{c}^{'})^2Q_{c} + 10Q_{c}^2Q_{c}^{''} + 6Q_{c}^5=0,\nonu\\%
%&\nonu\\%D[v[x, t], {x, 4}] + 10*D[v[x, t], {x, 1}]^2*v[x, t] +  10 v[x, t]^2 D[v[x, t], {x, 2}] + 6 v[x, t]^5 - k1^4*v[x, t]
&Q_{c}^{(vi)} - c^3Q_{c} + f_7(Q_c) = 0,
%42Q_{c}(Q_{c}^{''})^2 + 56 Q_{c}Q_{c}^{'}Q_{c}^{'''} + 14 Q_{c}^2Q_{c}^{(iv)} + 70(Q_{c}^{'})^2Q_{c}^{''}
%+ 70Q_{c}^{4}Q_{c}^{''} + 140Q_{c}^{3}(Q_{c}^{'})^2+ 20 Q_{c}^{7} =0.
%-k1^6*u7[x, t] + 20 u7[x, t]^7 +  70 (2 u7[x, t]^3 D[u7[x, t], {x}]^2
%-k1^6*u7[x, t] + 20 u7[x, t]^7 +  70 (2 u7[x, t]^3 D[u7[x, t], {x}]^2 +    u7[x, t]^4 D[u7[x, t], {x, 2}] +     D[u7[x, t], x]^2 D[u7[x, t], {x, 2}]) +
% 14*D[u7[x, t], {x, 4}]*(u7[x, t])^2 +  56 u7[x, t] D[u7[x, t], x] D[u7[x, t], {x, 3}] +  42 u7[x, t] D[u7[x, t], {x, 2}]^2 + D[u7[x, t], {x, 6}]
\end{align}
\medskip\noindent
and
\begin{align}\label{9thODE}
%Q_c'' -c\, Q_c + 2Q_c^3=0, \quad Q_c>0, , \quad Q_{c}>0, \quad Q_{c}\in H^1(\R).\\
%& Q_{c}^{iv} -c^2\, Q_{c} + 10(Q_{c}^{'})^2Q_{c} + 10Q_{c}^2Q_{c}^{''} + 6Q_{c}^5=0,\nonu\\%
%&\nonu\\%D[v[x, t], {x, 4}] + 10*D[v[x, t], {x, 1}]^2*v[x, t] +  10 v[x, t]^2 D[v[x, t], {x, 2}] + 6 v[x, t]^5 - k1^4*v[x, t]
&Q_{c}^{(viii)} - c^4Q_c + f_9(Q_c) = 0.
%\dots - c^3Q_{c} + 42Q_{c}(Q_{c}^{''})^2 + 56 Q_{c}Q_{c}^{'}Q_{c}^{'''} + 14 Q_{c}^2Q_{c}^{(iv)} + 70(Q_{c}^{'})^2Q_{c}^{''}
%+ 70Q_{c}^{4}Q_{c}^{''} + 140Q_{c}^{3}(Q_{c}^{'})^2+ 20 Q_{c}^{7} =0.
%-k1^6*u7[x, t] + 20 u7[x, t]^7 +  70 (2 u7[x, t]^3 D[u7[x, t], {x}]^2
%-k1^6*u7[x, t] + 20 u7[x, t]^7 +  70 (2 u7[x, t]^3 D[u7[x, t], {x}]^2 +    u7[x, t]^4 D[u7[x, t], {x, 2}] +     D[u7[x, t], x]^2 D[u7[x, t], {x, 2}]) +
% 14*D[u7[x, t], {x, 4}]*(u7[x, t])^2 +  56 u7[x, t] D[u7[x, t], x] D[u7[x, t], {x, 3}] +  42 u7[x, t] D[u7[x, t], {x, 2}]^2 + D[u7[x, t], {x, 6}]
\end{align}
Instead integrating directly in space \eqref{5mkdv}, \eqref{7mkdv} and \eqref{9mkdv}, another way to check the validity of \eqref{5thODE}, \eqref{7thODE} and \eqref{9thODE} is by using the
lowest order nonlinear stationary elliptic  equation \eqref{eqQc} satisfied by all higher order mKdV solitons. For instance, in the case of \eqref{5thODE},
we just
substitute and obtain:
\begin{align}\label{eqQc2}
%Q_c'' -c\, Q_c + 2Q_c^3=0, \quad Q_c>0, , \quad Q_{c}>0, \quad Q_{c}\in H^1(\R).\\(cQ_{c}^2-Q_{c}^4)
&Q_{c}^{(iv)} -c^2\, Q_{c} + f_5(Q_c)\nonu\\
&= Q_{c}^{(iv)} -c^2\, Q_{c} + 10(Q_{c}^{'})^2Q_{c} + 10Q_{c}^2Q_{c}^{''} + 6Q_{c}^5\nonu\\%
&= (cQ_{c}-2Q_{c}^3)^{''} - c^2Q_{c} + 10 (Q_{c}^{'})^2Q_{c} + 10Q_{c}^2(cQ_{c} -2Q_{c}^3) + 6Q_{c}^5\nonu\\
&= cQ_{c}^{''} -12Q_{c}(Q_{c}^{'})^2-6Q_{c}^2Q_{c}^{''} - c^2Q_{c} + 10 (Q_{c}^{'})^2Q_{c} + 10Q_{c}^2Q_{c}^{''} + 6Q_{c}^5\nonu\\
&=c(cQ_{c}-2Q_{c}^3)-2Q_{c}(cQ_{c}^2-Q_{c}^4)-6Q_{c}^2(cQ_{c}-2Q_{c}^3) - c^2Q_{c} + 10Q_{c}^2(cQ_{c}-2Q_{c}^3) + 6Q_{c}^5=0.
\end{align}
\noindent
The proof for the other higher order nonlinear identities \eqref{7thODE} and \eqref{9thODE} follows in the same way. Note moreover that the second order elliptic equation \eqref{eqQc} satisfied by all higher order mKdV solitons is deeply related to the variational meaning of the soliton solution.  To be more precise,
it is well-known that some of the (first) standard conservation laws of  5th, 7th and 9th-mKdV equations are the \emph{mass}
%these higher order mKdV equations are the \emph{mass}
%\be\label{M1}
%M_0[u](t)  :=  \frac 12 \int_\R u(t,x)^2dx = M_0[u](0),
%\ee
\begin{eqnarray}\label{M1cor} M[u](t)  :=  \frac 12\int_\R u^2(t,x)dx =
M[u](0), \end{eqnarray} the \emph{energy}

\be\label{E1}
%E[u](t)  :=   \int_\R (\frac 12  u_{xx}^2 -5u^2u_x^2 + u^6)(t,x)dx = E[u](0),
E[u](t)  := \frac 12 \int_\R\left( u_x^2 - u^4\right)(t,x)dx = E[u](0),
\ee

\medskip
\medskip \noindent 
and the \emph{higher order energies}, defined respectively in $H^2(\R)$ 

\be\label{E5} E_5[u](t)  :=   \int_\R\left(\frac 12u_{xx}^2 -5u^2u_x^2 + u^6\right)(t,x)dx = E_5[u](0),\\
%\displaystyle\left(\frac 12  u_{xx}^2 -5u^2u_x^2 + u^6\right)(t,x)dx = E_5[u](0),\\
%E[u](t)  :=  \frac 12 \int_\R u_x^2(t,x)dx -\frac 14 \int_\R u^4(t,x)dx = E[u](0).
\ee
\medskip
\medskip\noindent 
in $H^3(\R)$% and $\mathfrak{C}$
\be\label{E7} E_7[u](t)  :=   \int_\R\displaystyle\left(\frac 12 u_{3x}^2 + \frac 72 u_x^4 -  7u^2u_{xx}^2 + 35u^4u_x^2 - \frac 52 u^8\right)(t,x)dx= E_7[u](0), \ee 
\medskip
\noindent
and in $H^4(\R)$% and $\mathfrak{C}$
\be\label{E9} E_9[u](t)  :=   \int_\R\displaystyle\left(\frac 12u_{4x}^2 -9u^2u_{3x}^2 + 20uu_{xx}^3 + 51u_x^2u_{xx}^2 + 63u^4u_{xx}^2 -133u^2u_x^4 - 210u^6u_x^2 + 7u^{10}\right)(t,x)dx= E_9[u](0). \ee \noindent

\ms Using the \emph{lowest} order conserved quantities (i.e., mass and energy  \eqref{M1cor}-\eqref{E1}), the variational structure of
any higher order mKdV soliton \eqref{fexplicita2} can be characterized as follows: there exists a well-defined \emph{Lyapunov
functional}, \emph{invariant in time} and such that any higher order mKdV soliton $Q_{c}$ \eqref{fexplicita2} is an \emph{extremal point}. 
Moreover, it is a global minimizer under fixed mass. For the 5th, 7th and 9th-mKdV cases, this functional is given
by (see \cite{Benj} for the mKdV case) 

\be\label{H0mk} \mathcal{H}_0[u](t) = E[u](t) + c \, M[u](t), \ee

\medskip
\medskip \noindent 
where
$c>0$ is the scaling of the solitary wave \eqref{fexplicita2}, and $M[u]$, $E[u]$ are given in (\ref{M1cor}) and (\ref{E1}). Indeed, it is easy to see that for
any small perturbation $z(t)\in H^1(\R)$, 

\be\label{Expa1mk}
\mathcal{H}_0[Q_{c}+z](t)  =  \mathcal{H}_0[Q_{c}] - \int_\R z(Q_{c}''-cQ_{c} + 2Q_{c}^3) +  O(\|z(t)\|_{H^1}^2). \ee 

\medskip
\medskip \noindent 
The zero order term  above is independent of time, and the first order term in $z$ is zero from (\ref{eqQc}), which it implies the critical character of $Q_{c}$.

\medskip

Note that by using higher order conservation laws \eqref{E5} and \eqref{E7}, and therefore \emph{higher order} Lyapunov functionals,
we are also able to characterize 5th, 7th and 9th-mKdV solitons \eqref{fexplicita2} as extremal points of these \emph{higher order}
functionals. More precisely, for instance, in the 5th-mKdV case, and using the quantities $M[u]$, $E_5[u]$ given in (\ref{M1cor}) and (\ref{E5}),
this functional is explicitly given, for any $~c>0,$ by 

\be\label{H5mk}
\mathcal{H}_5[u](t) = E_5[u](t) - c^2 \, M[u](t).
\ee

\ms \noindent 
For the 7th-mKdV case, using the quantities $M[u]$, $E_7[u]$ given in (\ref{M1cor}) and (\ref{E7}), we get

\be\label{H7mk}
\mathcal{H}_7[u](t) = E_7[u](t) + c^3 \, M[u](t),
\ee

\medskip
\medskip \noindent 
and finally for the 9th-mKdV case, using the quantities $M[u]$, $E_9[u]$ given in (\ref{M1cor}) and (\ref{E9}), we get

\be\label{H9mk}
\mathcal{H}_9[u](t) = E_9[u](t) - c^4 \, M[u](t).
\ee

\medskip
\medskip \noindent 
In fact,  it is easy to see that for any small $z(t)\in H^2(\R)~~(\text{and}~~H^3(\R),~H^4(\R)~~~\text{respectively})$, 

\be\label{Expa5mk}
\mathcal{H}_5[Q_{c}+z](t)  =  \mathcal{H}_5[Q_{c}] + \int_\R z\Big(Q_{c}^{(iv)} -c^2\, Q_{c} + f_5(Q_c)\Big) +  O(\|z(t)\|_{H^2}^2), \ee
%+ 10(Q_{c}^{'})^2Q_{c} + 10Q_{c}^2Q_{c}^{''} + 6Q_{c}^5\Big) +  O(\|z(t)\|_{H^2}^2), \ee
\ms
\noindent 
\begin{align}\label{Expa7mk}
&\mathcal{H}_7[Q_{c}+z](t)  =  \mathcal{H}_7[Q_{c}] - \int_\R z\Big(Q_{c}^{(vi)} - c^3Q_{c} + f_7(Q_c)\Big) + O(\|z(t)\|_{H^2}^2),
%42Q_{c}(Q_{c}^{''})^2 + 56 Q_{c}Q_{c}^{'}Q_{c}^{'''} + 14 Q_{c}^2Q_{c}^{(iv)} + 70(Q_{c}^{'})^2Q_{c}^{''}
%+ 70Q_{c}^{4}Q_{c}^{''}\nonumber\\
%&+ 140Q_{c}^{3}(Q_{c}^{'})^2+ 20 Q_{c}^{7}\Big) +
%O(\|z(t)\|_{H^2}^2).
\end{align}
\noindent
and
\begin{align}\label{Expa9mk}
&\mathcal{H}_9[Q_{c}+z](t)  =  \mathcal{H}_9[Q_{c}] + \int_\R z\Big(Q_{c}^{(viii)} - c^4Q_{c} + f_9(Q_c)\Big) + O(\|z(t)\|_{H^2}^2).
%42Q_{c}(Q_{c}^{''})^2 + 56 Q_{c}Q_{c}^{'}Q_{c}^{'''} + 14 Q_{c}^2Q_{c}^{(iv)} + 70(Q_{c}^{'})^2Q_{c}^{''}
%+ 70Q_{c}^{4}Q_{c}^{''}\nonumber\\
%&+ 140Q_{c}^{3}(Q_{c}^{'})^2+ 20 Q_{c}^{7}\Big) +
%O(\|z(t)\|_{H^2}^2).
\end{align}
\noindent
In all cases, the zero order term  is independent of time, and  the first order term in $z$ is zero from (\ref{5thODE}),  \eqref{7thODE} and \eqref{9thODE}. 
Finally, and from the functionals \eqref{H5mk}-\eqref{H9mk} above, we \emph{conjecture} that the following Lyapunov functional (here we identify $E_3\equiv E$)

\be\label{HHiermk}
\mathcal{H}_{2n+1}[u](t) = E_{2n+1}[u](t) + (-1)^{n+1} c^n \, M[u](t),\quad n\in\N^+,
\ee

\medskip
\noindent
generates the associated nonlinear ODE 

\be\label{nthODE}
 Q_{c}^{(2n)} -c^n\, Q_{c} + f_{2n+1}(Q_c) =0,\quad n\in\N^+,
%10(Q_{c}^{'})^2Q_{c} + 10Q_{c}^2Q_{c}^{''} + 6Q_{c}^5=0,
%E[u](t)  :=  \frac 12 \int_\R u_x^2(t,x)dx -\frac 14 \int_\R u^4(t,x)dx = E[u](0).
\ee

\medskip
\noindent
satisfied by any soliton solution of the corresponding member of the \emph{focusing} mKdV hierarchy (see \cite{AleCar}).

\subsection{Breathers in 5th, 7th and 9th order mKdV equations}
Beside these soliton solutions of  5th, 7th and 9th-mKdV equations \eqref{5mkdv}-\eqref{7mkdv}-\eqref{9mkdv}, it is possible to find another big set of explicit and oscillatory solutions,
known in the physical and mathematical literature as the \emph{breather} solution, and which is a  spatially localized, 
and periodic in time, up to translations, real function.
% 
% Although there is no universal definition for a breather, we remember here the following convention introduced in \cite{AMP1},
% %that will match the  mKdV, Gardner and also sine-Gordon cases.
% 
% \begin{defn}[Aperiodic breather]\label{DEF_B}
% We say that $B=B(t,x)$ is a breather solution for a particular one-dimensional dispersive equation
% if there are $T>0$ and $L=L(T)\in \R$ such that, for all $t\in \R$ and $x\in \R$, one has
% \be\label{B_def}
% B(t+T,x) = B(t,x-L),
% \ee
% and moreover, the infimum among times $T>0$ such that property \eqref{B_def} is satisfied for such a time $T$ is uniformly positive in space.
% \end{defn}
% %
% 
% \begin{rem}
% Observe that the last condition ensures that solitons (and multisolitons) are not breathers, since e.g. $Q_{c}(x-c (t+T)) = Q_{c}(x-L -ct)$ for $L:=cT$ but $T$ can be any real-valued
% time.
% %\footnote{In the case of NLS equations and their solitons, Definition \ref{DEF_B} includes them because of the $U(1)$ invariance.}
% \end{rem}

\medskip

For the 5th, 7th and 9th-mKdV equations \eqref{5mkdv}-\eqref{7mkdv}-\eqref{9mkdv}, the breather solution in the line can be obtained by using different methods
 (e.g. Inverse Scattering, Hirota method. See \cite{Mat, Mat1} for further details). Particularly we use here a matching method to find these breather
 solutions, i.e. proposing a well known ansatz, with  speeds as free parameters to be determined in order to define a solution. Note that the same
 procedure can be used to obtain periodic breather solutions of the  5th, 7th and 9th-mKdV equations.

%BREATHER GARDNER MU CUADRATICO
\begin{defn}[5th, 7th and 9th-mKdV breathers]\label{579breather} Let $\al, \bt >0$
and $x_1,x_2\in \R$. The real-valued  breather solution of the 5th, 7th and 9th-mKdV equations \eqref{5mkdv}-\eqref{7mkdv}-\eqref{9mkdv} is given explicitly by the formula
\be\label{579Bre}B\equiv B_{\al, \bt}(t,x;x_1,x_2)   %:=   \partial_x \tilde B_\mu
 :=   2\partial_x\Bigg[\arctan\Big(\frac{\bt}{\al}\frac{\sin(\al y_1)}{\cosh(\bt y_2)}\Big)\Bigg],
\ee
with $y_1$ and $y_2$
\be\label{y1y2GE}
\begin{aligned}
&y_1 = x+ \delta_i t + x_1, \quad y_2 = x+ \ga_i t +
x_2,~~~i=5,7,9\end{aligned}\ee
\noindent and with velocities
$(\delta_5,\ga_5)$ in the 5th order case

\be\label{speeds579}
\begin{aligned}
&\delta_5 := -\al^4+10\al^2\bt^2-5\bt^4,\quad\ga_5 :=-\bt^4+10\al^2\bt^2-5\al^4,\\
\end{aligned}\ee
%&\text{and}\quad\\
\noindent  $(\delta_7,\ga_7)$ in the 7th order case
\be\label{speeds579-1}
\begin{aligned}
&\delta_7 := \al^6-21\al^4\bt^2+35\al^2\bt^4-7\bt^6,\quad\ga_7 :=-\bt^6+21\al^2\bt^4-35\al^4\bt^2+7\al^6,
%&\delta := -\al^4+10\al^2\bt^2-5\bt^4,\quad\text{and}\quad\ga :=-\bt^4+10\al^2\bt^2-5\al^4.
\end{aligned}\ee
\noindent and  $(\delta_9,\ga_9)$ in the 9th order case
\be\label{speeds579-2}
\begin{aligned}
&\delta_9 := -\al^8 + 36\al^6\bt^2 - 126\al^4\bt^4 + 84\al^3\bt^6-9\bt^8,\quad\ga_9 :=-\bt^8+36\al^2\bt^6-126\al^4\bt^4 + 84\al^6\bt^2-9\al^8.\\
%&\delta := -\al^4+10\al^2\bt^2-5\bt^4,\quad\text{and}\quad\ga :=-\bt^4+10\al^2\bt^2-5\al^4.
\end{aligned}\ee
\end{defn}
%Gardner breather solutions have also been used  \cite{Ale1} to prove the ill-posedness for the Gardner equation \eqref{GE} in the line,
%showing that solutions cannot depend in an uniformly continuous form on their initial data in the Sobolev spaces $H^s(\R)$  for $s < 1/4$.

\ms
% \begin{rem}
% This is a four-parametric solution, with two scalings ($\al,\bt$) and two shift translations ($x_1,x_2$).
% Note moreover that from this formula one has, for any $k\in \Z$,
% \be\label{dege}
%  B_{\al,\bt} \displaystyle\left(t,x; x_1 + \frac{k\pi}{\al}, x_2\right)  =  (-1)^k B_{\al,\bt} (t,x; x_1, x_2),
% \ee which are also solutions of (\ref{5mkdv})-(\ref{7mkdv}). This
% identity reveals the periodic character of the first translation
% parameter $x_1$.
% \end{rem}

\begin{rem}
 Observe that  breather solutions for 5th, 7th and 9th order mKdV equations have the same functional expression as the \emph{classical}
 mKdV breather solution \cite[Def.1.1]{AM}
 \be\label{Bremkdv}
  B\equiv B_{\al, \bt}(t,x;x_1,x_2)   %:=   \partial_x \tilde B_\mu
 :=   2\partial_x\Bigg[\arctan\Big(\frac{\bt}{\al}\frac{\sin(\al y_1)}{\cosh(\bt y_2)}\Big)\Bigg], \eeq
\noindent
with \quad $y_1 = x+ \delta t + x_1,~~ y_2 = x+ \ga t +x_2,$
\noindent
and  velocities $\delta=\al^2-3\bt^2,~~\ga=3\al^2-\bt^2$, and in fact only differing in speeds  \eqref{speeds579}-\eqref{speeds579-2}.

\end{rem}

\begin{rem}
 Finally be aware that these 5th, 7th and 9th breather solutions \eqref{579Bre} in $\R$ could be used to re-approach
 the ill-posedness of the Cauchy problem for 5th, 7th and 9th-mKdV equations \eqref{5mkdv}-\eqref{7mkdv} and \eqref{9mkdv},
 in the same way they were used by Kenig-Ponce and Vega \cite{KPV2} and Alejo \cite{Ale1}, to show a failure  of
 the flow map associated to some nonlinear dispersive equations to be uniformly continuous. This
 procedure could afford a complementary proof to the previous works on the ill-posedness of these higher order equations presented by  Kwon \cite{Kwon}
 and  Gr\"unrock \cite{Gru}. 
%Note finally, that in the periodic setting for the  5th, 7th and 9th-mKdV equations \eqref{5mkdv} and \eqref{7mkdv},
%  we could follow the same reasoning,
%  but instead using the periodic breathers \eqref{57Bper} to handle ill-posedness questions on these equations.
%  See appendix \ref{Sect5} for further reading on these periodic
%  breathers of higher order 5th, 7th and 9th equations.

\end{rem}

\medskip

One of the main results of this work will be to prove that, exactly as it happens with all 5th, 7th and 9th soliton solutions \eqref{fexplicita2} which satisfy the same
nonlinear elliptic equation \eqref{eqQc},  \emph{breather} solutions \eqref{579Bre} of the 5th, 7th and 9th mKdV equations  satisfy the \emph{same} nonlinear
fourth order stationary elliptic equation. Namely

\begin{thm}\label{MT0} Any 5th, 7th or 9th mKdV breather $B$ satisfies the same fourth order stationary elliptic equation than the \emph{classical} mKdV breather, namely
\[
 B_{4x} + 10 BB_x^2 + 10B^2 B_{xx} + 6B^5 -2(\bt^2 -\al^2) (B_{xx} + 2B^3)  + (\al^2 +\bt^2)^2B =0.\]

\end{thm}

This fact can be interpreted as if all mKdV breathers and higher order mKdV breathers are characterized
by the same elliptic equation, in a similar way as it was showed for the KdV equation by Lax \cite{LAX1}. 
Moreover, and as second main result in this paper, we give a  positive answer to the question of breathers stability for these higher order mKdV equations. 

\begin{thm}\label{MT1} 5th, 7th and 9th mKdV breathers are orbitally stable in the $H^2$-topology. 
\end{thm}

A more detailed version of this result is given in Theorem \ref{T1gardner}. 
As we have already shown, we need the space $H^2$  by a regularity argument and through the variational characterization that we obtain
of these breather solutions of higher order mKdV equations. 

% Even more,  this space comes from the fact that breather structures
% are \emph{bound states}, which means that there is no mass decoupling as time evolves. 

\bigskip

\subsection{Organization of this paper}
In Sect.\ref{2}  we present some higher order  nonlinear identities adapted to 5th, 7th and 9th-mKdV breathers. Furthermore, we prove that any 5th, 7th or 9th-mKdV breather solutions satisfy a fourth order nonlinear ODE, which characterizes them. Sect.\ref{3} is devoted to collect and list the properties of a linearized operator associated to these higher order breather solutions. In Sect.\ref{4} we introduce a suitable $H^2$-Lyapunov functional for  higher order mKdV equations (\ref{5mkdv}), (\ref{7mkdv}) and (\ref{9mkdv}).
Finally, in Sect.\ref{5} we present a detailed version of Theorem \ref{T1gardner}.

\bigskip
 \noindent

{\bf Acknowledgments.}
We would like to thank to professors C. Mu\~noz and C. Kwak for richful discussions and comments on a previous version.

\section{ Higher order nonlinear identities}\label{2}

\medskip
 The aim of this section is to show explicit nonlinear identities satisfied by any 5th, 7th or 9th-mKdV breathers.

First of all, consider the two directions associated to spatial translations. Let $B_{\al,\beta}$ as introduced in \eqref{579Bre}. Then we define

\be\label{B12}
 B_1(t ; x_1,x_2) := \partial_{x_1} B_{\al,\bt}(t ; x_1, x_2)\quad \hbox{ and } \quad  B_2(t ; x_1,x_2): =\partial_{x_2} B_{\al,\bt}(t ;  x_1, x_2).
\ee

\bigskip
\noindent
It is clear that, for all $t\in \R,$ and  $\al,\bt$ as in definition \eqref{579Bre} and $x_1,x_2\in \R$, both $B_1$ and $B_2$
are real-valued, exponentially decreasing in space, functions in the  Schwartz class. Moreover, it is not difficult to see that they are \emph{linearly independent} as functions of
the $x$-variable, for all time $t$ fixed. We also define the \emph{partial} mass associated to any 5th, 7th or 9th-mKdV breather $B$ \eqref{579Bre} as ($G=\frac{\bt}{\al}\sin(\al y_1),~F=\cosh(\bt y_2)$):

\begin{align}\label{partialMass}
 \mathcal M(t,x)\equiv\mathcal M_{\al,\bt}(t,x)  &:= \frac 12\int_{-\infty}^x B^2(t,s; x_1,x_2)ds ~= \bt +\frac 12 \partial_x \log(G^2 + F^2)(t,x).
% \frac{2\bt [ \al^2 +\bt^2 +\al\bt \sin(2\al y_1) -\bt^2 \cos(2\al y_1) +\al^2 (\sinh(2\bt y_2) + \cosh(2\bt y_2))]}{\al^2 +\bt^2  +\al^2 \cosh (2\bt y_2) - \bt^2 \cos(2\al y_1)},
% DEFINIR F y G en algún momento\frac{\bt}{\al}\frac{\sin(\al y_1)}{\cosh(\bt y_2)
\end{align}
\ms
\noindent
Here we have used that $B^2 = \frac 12 \partial_x^2 \log(G^2 + F^2)(t,x)$. See \cite{AM1} for further details. Finally, let consider $\tilde B=\tilde B_{\al,\bt}$  as the following $L^{\infty}$-function associated to mKdV breathers:
\bigskip
\noindent
%\item  $B =\tilde B_{x}$, with given by %the smooth $L^\infty$-function
\be\label{5tBmkdv}
\tilde B(t,x) := 2\arctan \Big(\frac{\bt}{\al}\frac{\sin(\al y_1)}{\cosh(\bt y_2)}\Big).
\ee

\bigskip
\noindent
The following nonlinear identities are satisfied by 5th, 7th and 9th-mKdV breathers:
%for 5th-mKdV breathers:

\medskip\bigskip

\begin{lem}\label{5Id1mkdv} We have for all $t\in \R,$ and  $\al,\bt>0,$ the following identities. Let $B=B_{\al,\bt}$ be any  5th, 7th or  9th-mKdV breather solution of the form \eqref{579Bre} as it corresponds. Then\\
\begin{enumerate}
\item For any fixed $t\in \R$, we have $ (\tilde B)_t$ well-defined in the Schwartz class, satisfiying respectively for the  5th, 7th or  9th-mKdV equations  that

 \be\label{52ndmkdvnew}
\tilde B_{t} + B_{2nx} + f_{2n+1}(B)=0,\quad \text{with}\quad n=2,~3,~4.
\ee
\bigskip
% 
% \begin{enumerate}
%  \item\emph{The 5th order case:}
%  \be\label{52ndmkdv}
% \tilde B_{t} + B_{4x} + f_5(B)=0.
% %10BB_{x}^2 + 10B^2B_{xx} + 6B^5= 0.
% \ee
%  \item \emph{The 7th order case:}
%  \be\label{2ndmkdv}
% \tilde B_{t} + B_{6x} + f_7(B) = 0.
% %14B^2B_{4x} +56BB_xB_{3x} + 42BB_{xx}^2 + 70B_x^2B_{xx} + 70B^4B_{xx} + 140B^3B_x^2 + 20B^7= 0.
% \ee
%  
%  \item \emph{The 9th order case:}
%  \be\label{2ndmkdv9}
% \tilde B_{t} + B_{8x} + f_9(B) =0.
% %14B^2B_{4x} +56BB_xB_{3x} + 42BB_{xx}^2 + 70B_x^2B_{xx} + 70B^4B_{xx} + 140B^3B_x^2 + 20B^7= 0.
% \ee
%  \end{enumerate}
%  
 \item Let $\mathcal M$ be defined by (\ref{partialMass}). Then
 
 \medskip
 
 \begin{enumerate}
 \item\emph{The 5th order case:}
 \be\label{5Firstmkdv}\begin{aligned}
&B_{xx}^2 - 2B \tilde B_{t} + 2(\mathcal M)_t - 2B^6 - 2B_xB_{xxx}  - 10B^2B_{x}^2 =0.
\end{aligned}
\ee
 \item \emph{The 7th order case:}
\be\label{Firstmkdv}\begin{aligned}
B_{3x}^2 &+ 2B \tilde B_{t} - 2(\mathcal M)_t + 5B^8 + 2B_xB_{5x} -2B_{xx}^2B_{4x} \\
&+ 28B^2B_xB_{3x}  -14B^2B_{xx}^2 + 56BB_x^2B_{xx} + 7B_{x}^4  +  70B^4B_{x}^2  =0.
\end{aligned}
\ee
 
 \item \emph{The 9th order case:}
\be\label{Firstmkdv9}\begin{aligned}
B_{4x}^2  &- 2B \tilde B_{t} + 2(\mathcal M)_t - 2B_{7x}B_x + 2B_{6x}B_{xx} -2B_{5x}B_{3x} + F[B] =0,\\
&~F[B]:=-2\int_{-\infty}^xf_9(B)(s)B_sds.
\end{aligned}
\ee
 \end{enumerate}
 \end{enumerate}
\end{lem} 
\begin{proof}
In the 5th case, the first item (\ref{52ndmkdvnew}) is a consequence of (\ref{5tBmkdv}) and a convenient integration in space
(from $-\infty$ to $x$). To obtain (\ref{5Firstmkdv}) we multiply (\ref{52ndmkdvnew}), when $n=2,$ by $B_{x}$ and integrate in space in the same region. The proofs in the 7th and 9th order cases follow similar steps as in the 5th order case.
\end{proof}

We compute now the \emph{higher order} energies \eqref{E5}, \eqref{E7} and \eqref{E9} of any higher order breather solution of \eqref{5mkdv},  \eqref{7mkdv} and \eqref{9mkdv} equations.

\begin{lem}\label{ME7} Let $B=B_{\al,\bt}$ be any 5th, 7th or 9th order mKdV breather solutions respectively, for $\al, \bt$ as in definition \eqref{579Bre}.
Then the higher order energies \eqref{E5}, \eqref{E7} and \eqref{E9}  of a 5th, 7th and 9th-mKdV breather $B$ are respectively

\be\label{EnergyB5}
E_5[B] :=  -\frac{2}{5}\bt\ga_5,\quad E_7[B] :=  \frac{2}{7}\bt\ga_7,\quad \text{and}\quad E_9[B] :=  -\frac{2}{9}\bt\ga_9,\\
%= 2\bt \ga |E[Q]|.
\ee
\ms
\noindent
%and
%\be\label{EnergyB5}
%E_7[B] :=  \frac{4}{7}\bt\ga_7,
%= 2\bt \ga |E[Q]|.
%\ee
%\noindent

with $\ga_5,~\ga_7,~\ga_9$ given in \eqref{speeds579}-\eqref{speeds579-1}-\eqref{speeds579-2}.
\end{lem}
\ms
\noindent
\begin{rem}\label{obs1}
Note that as it happens with the classical mKdV breather solution $B$,  where $E[B] :=  \frac{2}{3}\bt\ga$ (see \cite[Lemma 2.4]{AM}),
the sign of the higher order energies $E_5,~ E_7,~E_9$ is driven by a nonlinear balance among the different terms depending on scalings $\al,\bt$. 
\end{rem}

\begin{rem}
From the above Lemma, we \emph{conjecture} that for any $(2n+1)$-order mKdV breather $B$, its $(2n+1)$-order energy is given by
\be\label{conjecE}
E_{2n+1}[B] (t)= (-1)^{n+1}\frac{2\bt}{2n+1}\ga_{2n+1},~~n\in\N,
\ee
\noindent
and with
\[
 \gamma_{2n+1} := \sum_{j=0}^{n}(-1)^j\frac{(2n+1)!}{(2j)!(2n+1-2j)!}\al^{2j}\bt^{2(n-j)},\quad\hspace{1cm}~~~~n\in\N.
\]

\medskip
%Therefore, the associated energy for the 9th-mKdV breather $B$ \eqref{579Bre}, should be $E_9[B] :=  -\frac{2}{9}\bt\ga_9$.
%\[
% \gamma_p := \sum_{j=0}^{p}(-1)^j\frac{(2p+1)!}{(2j)!(2p+1-2j)!}\al^{2j}\bt^{2(p-j)+1},\quad\hspace{1cm}~~~~p=1,2,3,\dots,\\
%\]

%Note that we could follow the approach by Lax \cite[pp. 479--481]{LAX1} to obtain reduced expressions for
%the mass and energy of a breather solution. However, the resulting terms are actually harder to manage than our  direct approach.
\end{rem}

\begin{proof}{(of Lemma \ref{ME7})}\newline

We start with the 5th order case. First of all,  let us prove the following reduction
\be\label{red2}
E_5[B] (t)= -\frac 15\int_\R \Big((\mathcal M)_t(t,x)\Big)dx.
\ee
Indeed,  we multiply (\ref{52ndmkdvnew}) by $B$ and integrate in space: we get
\[
\int_\R B_{xx}^2  = \int_\R 20B^2B_x^2 - 6B^6 -B \tilde B_{t}.
\]
On the other hand, integrating (\ref{5Firstmkdv}),
\[
\int_\R B_{xx}^2 =   \frac 23 \int_\R B^6 + \frac 23\int_\R B\tilde B_{t} - \frac 23 \int_\R (\mathcal M)_t + \frac{10}{3}\int_\R B^2B_x^2.
\]
From these two identities, we get
\[
 \int_\R B^6 = \frac{1}{10} \int_\R (\mathcal M)_t -\frac 14\int_\R B \tilde B_{t} + \frac{5}{2} \int_\R B^2B_x^2,
\]
and therefore
\[
\int_\R B_{xx}^2  = -\frac 35 \int_\R (\mathcal M)_t +
\frac{15}{3}\int_\R B^2B_x^2 +\frac 12 \int_\R B \tilde B_{t}.
\]
Finally, substituting the last two identities into (\ref{E5}), we get (\ref{red2}), as desired. Proceeding in the same way,  in the
7th and 9th order cases we obtain the corresponding simplications

\be\label{red279}
E_7[B] (t)= \frac 17\int_\R \Big((\mathcal M)_t(t,x)\Big)dx,\quad E_9[B] (t)= \frac 19\int_\R \Big((\mathcal M)_t(t,x)\Big)dx.
\ee

\medskip

Now we prove (\ref{EnergyB5}).  From (\ref{partialMass}), we have that
% 
% \[
%  \mathcal M(t,x)= \bt + \frac12\partial_x\log(G^2+F^2),
% \]
% \noindent and hence,
\[
 \mathcal M_t(t,x)=  \frac12\partial_x\partial_t\log(G^2+F^2)(t,x).
\]
\noindent
Now substituting in the energy \eqref{red2}, remembering the identity \eqref{52ndmkdvnew} and the explicit expression for
 $\mathcal M[B]$ in \eqref{partialMass}, we get

\begin{align*}
   E_5[B] (t)&= -\frac 15\int_\R \Big((\mathcal M)_t(t,x)\Big)\ dx =
-\frac 15\frac12\int_\R \Big(\partial_x\partial_t\log(G^2+F^2) \Big)dx\\%B_{\mu,xx} + \tilde B_{\mu,t} + 3\mu B_\mu^2 + B_\mu^3
& = -\Big(\frac 15\frac12\partial_t\log(G^2+F^2)
\Big)|_{-\infty}^{+\infty} = -\frac{2}{5}\bt\ga_5.
\end{align*}
\noindent
For the 7th and 9th order cases, we proceed as above, but now using \eqref{52ndmkdvnew}, \eqref{Firstmkdv} and \eqref{Firstmkdv9}, and we get
\[
 E_7[B]=\frac{2}{7}\bt\ga_7,\quad\text{and}\quad E_9[B]=-\frac{2}{9}\bt\ga_9.
\]

\end{proof}

Note that since the profiles of 5th, 7th and 9th order mKdV breathers (solitons) agree with the expression of the classical
mKdV breather (soliton), and since the
energy $E$ \eqref{E1} is a conserved quantity for the mKdV and  5th, 7th and 9th higher order equations, when the lowest energy
$E$ \eqref{E1} is evaluated in these 5th, 7th and 9th higher order breathers we obtain in both cases the same value than the mKdV breather energy, $\frac{2}{3}\bt\ga$.
For the sake of simplicity and to understand that property, we remember here the relation \cite[(4.2),(4.4)]{AM1} in the case of low order
conserved quantities evaluated at breather solutions $B$ and at soliton solutions $Q_c$:

\be\label{r1mass}M[B] = 2 Re \Big[M[Q_{c}]|_{\sqrt{c}=\bt+i\al}\Big]\qquad\text{and}\qquad E[B] = 2 Re\Big[ E[Q_{c}]|_{\sqrt{c}=\bt+i\al}\Big].\ee
 \noindent
\medskip
% 
% \begin{cor}\label{WCcor2} Let $B=B_{\al,\bt}$ be any 5th or 7th order mKdV breather. Then
% \be\label{LAB3}
% \partial_\al E[B] =4\al\bt,
% \quad \partial_\bt E[B] =2(\al^2-\bt^2).
% %12 \al\bt |E[Q]|>0, \quad \partial_\bt E[B] = 6(\al^2 -\bt^2) |E[Q]|.
% \ee
% \end{cor}
% 
% \begin{rem}
% Note that  condition $\al=\bt$ is equivalent to the identity $\partial_\bt E[B] =0.$ Higher order solitons do not satisfy this last indentity.
% \end{rem}

%\begin{rem}
%Note that the condition $\al=\bt$ is equivalent to the identity
%\[
% \partial_\bt E[B] =\frac{8\mu^2\bt^4}{\bt^4+\mu^4},
%\]
%\noindent
%which is always positive, on the contrary to the mKdV breather case ($\mu=0$). On the other hand the similar identity \eqref{WC}
%for the energy of Gardner solitons  can not vanish for any $\mu$.
%\end{rem}

\bigskip
%%%%%%%%%%%%%%%%%%%%%%
The next nontrivial identity for   5th-mKdV breathers \eqref{579Bre} will be useful in the proof of the
nonlinear stationary equation that they satisfy.

\begin{lem}\label{Iddificil}
Let $B=B_{\al,\bt}$ be any  5th-mKdV breather \eqref{579Bre}.
Then, for all $t\in \R$,
%\be\label{ide1nvbc}
% B_{xt} +  2(\mathcal M)_t B   = \Big(2(\bt^2 -\al^2)\Big)\tilde B_t  + (\al^2 +\bt^2)^2B.
%\ee

\be\label{ide1nvbc2}
 \tilde{B}_{t}    = (\al^2 +\bt^2)^2B - 2(\bt^2 -\al^2)(B_{xx} + 2B^3).
\ee
\end{lem}
\begin{proof}
We will use the following notation:

\begin{eqnarray}\label{breHN}
&& B:=2\partial_x\Big[\arctan\Big(\frac{\bt}{\al}\frac{\sin(\al y_1)}{\cosh(\bt y_2)}\Big)\Big]=\frac{H(t,x)}{N(t,x)} ,\nonu\\
&&\nonu\\
&&H:=H(t,x)=2\Big(\bt\al^2\cosh(\bt y_2)\cos(\al y_1) - \bt^2\al\sinh(\bt y_2)\sin(\al y_1)\Big),\nonu\\
&&\nonu\\
&&N:=N(t,x)=\al^2\cosh^2(\bt y_2)+\bt^2\sin^2(\al y_1),\nonu
\end{eqnarray}
\noindent
and from $\tilde{B}$ \eqref{5tBmkdv},
\begin{eqnarray}\label{brePN}
&& \tilde{B}_{t}:=2\partial_t\Big[\arctan\Big(\frac{\bt}{\al}\frac{\sin(\al y_1)}{\cosh(\bt y_2)}\Big)\Big]=\frac{P(t,x)}{N(t,x)} ,\nonu\\
&&\nonu\\
&&P:=P(t,x)=2\Big(\bt\al\delta_5\cosh(\bt y_2)\cos(\al y_1) -
\bt\al\gamma_5\sinh(\bt y_2)\sin(\al y_1)\Big),
%&&\nonu\\
%&&N:=N(t,x)=\al^2\cosh^2(\bt y_2)+\bt^2\sin^2(\al y_1),\nonuspeeds579
\end{eqnarray}
\noindent
with $\delta_5,\gamma_5$ as in \eqref{speeds579}. For the sake of simplicity, we are going to use the following notation:
\begin{align}\label{notacionN1}
%&  N[x,t]:=N= \al^2\cosh^2(\bt y_2)+\bt^2\sin^2(\al y_1),\\
&  N_1:=N_x=2\al\bt^2\cos(\al y_1)\sin(\al y_1) + 2\al^2\bt\cosh(\bt y_2)\sinh(\bt y_2),\\
&  N_2:=N_{xx}=2\al^2\bt^2(\cos^2(\al y_1)-\sin^2(\al y_1) +
\cosh^2(\bt y_2) + \sinh^2(\bt y_2)),
%&  N_3:=N_{xxx}=-8\al^3\bt^2\cos(\al y_1)\sin(\al y_1) + 8\al^2\bt^3\cosh(\bt y_2)\sinh(\bt y_2),\\
%&  N_4:=N_{4x}=8\al^2\bt^2(-\al^2\cos^2(\al y_1) + \al^2\sin^2(\al y_1) + \bt^2\cosh^2(\bt y_2) + \bt^2\sinh^2(\bt y_2)).
\end{align}
\medskip
and
\begin{align}
%&  H[x,t]:=H=2\Big(\bt\al^2\cosh(\bt y_2)\cos(\al y_1) - \bt^2\al\sinh(\bt y_2)\sin(\al y_1)\Big),\\
&  H_1:=H_x=-2\al\bt(\al^2+\bt^2)\cosh(\bt y_2)\sin(\al y_1),\\
&  H_2:=H_{xx}=-2\al\bt(\bt^2+\al^2)(\al\cosh(\bt y_2)\cos(\al y_1)
+\bt \sin(\al y_1)\sinh(\bt y_2)).\label{notacionHN}
%&- \bt(2\al\bt\cosh(\bt y_2)\cos(\al y_1) +(\bt^2-\al^2)\sin(\al y_1)\sinh(\bt y_2)),\\
%&  H_3:=H_{xxx}=2\al\bt((\al^4-\bt^4)\cosh(\bt y_2)\sin(\al y_1) - 2\al\bt(\al^2 + \bt^2)\cos(\al y_1)\sinh(\bt y_2)),\\
%&  H_4:=H_{4x}=2\al\bt((\al^5-2\al^3\bt^2-3\al\bt^4)\cosh(\bt y_2)\cos(\al y_1) \nonu\\
%&+ (3\al^4\bt+2\al^2\bt^3-\bt^5)\sin(\al y_1)\sinh(\bt y_2)).
\end{align}
\noindent
First of all, we start rewriting the following terms of the l.h.s. of \eqref{ide1nvbc2}:

\be\label{h1} B_{xx} +2B^3 = \frac{1}{N^3}\Big(2 H^3 + H_2 N^2 - 2 H_1N N_1 + 2 H N_1^2 - H N N_2\Big), \ee \noindent
\noindent
and hence, we have that

\be\label{identidad5th}
\begin{aligned}
& - \tilde B_t  - 2(\bt^2 -\al^2)(B_{xx} + 2B^3)  + (\al^2 +\bt^2)^2B=\frac{M_0}{N^3},
\end{aligned}
\ee
\noindent
with

\be\label{M0}
\begin{aligned}
&M_0:=-PN^2 + (\al^2 +\bt^2)^2HN^2 - 2(\bt^2 -\al^2) \Big(2H^3-2NH_1N_1+2HN_1^2+N^2H_2-HNN_2\Big).\\
%\Big(8 H^5 + 2 H^2 N (5 H_2 N - 22 H_1 N_1)+   2 H^3 (16 N_1^2 - 5 N N_2 + 2(\al^2 - \bt^2)N^2)\nonu\\
%&+   H \Big[24 N_1^4 - 36 N N_1^2 N_2
%+  2 N^2 (6 H_1^2 + 3 N_2^2 + 4 N_1 N_3 +  2(\al^2 - \bt^2) N_1^2 ) - N^3 (N_4 + 2  (\al^2 - \bt^2)N_2)\Big]\nonu\\
%& +   N [-24 H_1 N_1^3 + 12 N N_1 (H_2 N_1 + 2 H_1 N_2) +  N^3 (H_4 + 2(\al^2 - \bt^2)H_2) \nonu\\
%& - 2 N^2 (2 H_3 N_1 + 3 H_2 N_2 +  2 H_1 (N_3 + (\al^2 -
%\bt^2)N_1))]\Big).
\end{aligned}
\ee
\noindent
Indeed, we verify,  after substituting $P$ and $H's$ and $N's$ terms explicitly in
\eqref{M0} and having in mind basic trigonometric and hyperbolic identities, that
\be\label{sumaM0}
%\begin{aligned}
M_0=0,
%\end{aligned}
\ee
\noindent
and we conclude.
\end{proof}

We are ready now to present one of the most important results of this work, namely, we are going to show that in fact, breather solutions \eqref{579Bre}
of 5th, 7th and 9th-mKdV equations satisfy the same fourth order ODE satisfied by the \emph{classical} mKdV breather solution \eqref{Bremkdv} and it characterizes them. This
result means that this ODE identifies breather functions at different levels in the mKdV hierarchy, i.e. at the mKdV level and at 5th, 7th and 9th mKdV levels,
as being solutions of the same stationary fourth order ODE.

\begin{thm}\label{GBnvbc} Let $B= B_{\al,\bt}$ be any 5th, 7th or 9th-mKdV breather solution given in  \eqref{579Bre}.
Then, for any  fixed $t\in \R$, $B$ satisfies the same nonlinear stationary equation than the classical mKdV breather solution \eqref{Bremkdv}, namely

\bea\label{EcBnvbc} &G[B]:=B_{4x} + 10 BB_x^2 + 10B^2 B_{xx} + 6B^5 -2(\bt^2 -\al^2) (B_{xx} + 2B^3)  + (\al^2 +\bt^2)^2B =0.\ \eea

\end{thm}

\medskip

\begin{proof}%[Proof of Proposition \ref{GBnvbc}]
In the case of the 5th order breather, since by \eqref{52ndmkdvnew} the first four terms in \eqref{EcBnvbc} equal  $- \tilde B_t$ and using the above identity
\eqref{ide1nvbc2}, we simply get
\begin{align*}
& G[B] = - \tilde B_t  - 2(\bt^2 -\al^2)(B_{xx} + 2B^3)  + (\al^2 +\bt^2)^2B = 0.
%B_{xt} +  2(\mathcal M_{\al,\bt})_t B   = [2(\bt^2 -\al^2)+5\mu^2]\tilde B_t  + [(\al^2 +\bt^2)^2 +6\mu^2(\bt^2-\al^2 + \frac{3}{2}\mu^2)](B-\mu)
\end{align*}
%In the last line we have used (\ref{ide1nvbc2}).

%\medskip

The 7th and 9th order cases are more involved since we do not have at hand any identity like (\ref{ide1nvbc2}). Therefore, we first recast the l.h.s. of
\eqref{EcBnvbc}. Taking into account the r.h.s. of \eqref{v25mkdv}, we rewrite the first four terms in \eqref{EcBnvbc} and
 simplify the l.h.s. of \eqref{EcBnvbc}, as follows:

\begin{align}\label{7thODEv0}
& B_{4x} + 10 BB_x^2 + 10B^2 B_{xx} + 6B^5 -2(\bt^2 -\al^2) (B_{xx} + 2B^3)  +(\al^2 +\bt^2)^2B\nonu\\
& = \partial_{x}^2(B_{xx}+ 2B^3) - 2B(B_x^2 - 2BB_{xx} - 3B^4)-2(\bt^2 -\al^2) (B_{xx} + 2B^3)  +(\al^2 +\bt^2)^2B\nonu\\
& =\partial_{x}^2(B_{xx}+ 2B^3) - 2B([B_x^2 + B^4]- 2B[B_{xx} +2B^3])-2(\bt^2 -\al^2) (B_{xx} + 2B^3)  +(\al^2 +\bt^2)^2B\nonu\\
& = \partial_{x}^2(B_{xx}+ 2B^3) + (4B^2-2(\bt^2 -\al^2))(B_{xx} +2B^3) - 2B[B_x^2 + B^4]+(\al^2 +\bt^2)^2B.
%& = \partial_{x}^2(B_{xx}+ 2B^3) -2(\bt^2 -\al^2) (B_{xx} + 2B^3)  +(\al^2 +\bt^2)^2B.
%& = \partial_{x}^2(B_{xx}+ 2B^3) - (2BB_x^2 - 4B^2B_{xx} - 12B^5)-2(\bt^2 -\al^2) (B_{xx} + 2B^3)  +(\al^2 +\bt^2)^2B.
\end{align}
\noindent
Now, we prove directly that \eqref{7thODEv0} vanishes. Having in mind notation \eqref{breHN} and \eqref{notacionN1}-\eqref{notacionHN}, we
extend it considering the  following derivatives:

%\begin{eqnarray}\label{breHN}
%&& B:=2\partial_x\Big[\arctan\Big(\frac{\bt}{\al}\frac{\sin(\al y_1)}{\cosh(\bt y_2)}\Big)\Big]=\frac{H(t,x)}{N(t,x)} ,\nonu\\
%&&\nonu\\
%&&H:=H(t,x)=2\Big(\bt\al^2\cosh(\bt y_2)\cos(\al y_1) - \bt^2\al\sinh(\bt y_2)\sin(\al y_1)\Big),\nonu\\
%&&\nonu\\
%&&N:=N(t,x)=\al^2\cosh^2(\bt y_2)+\bt^2\sin^2(\al y_1).\nonu
%\end{eqnarray}

%and remembering that $\Delta=\al^2+\bt^2-2\mu^2$ and $e^{z} = \cosh[z] + \sinh[z]$, we have that:

%For the sake of simplicity, we are going to use the following notation:
\begin{align}\label{notacionN17}
%&  N[x,t]:=N= \al^2\cosh^2(\bt y_2)+\bt^2\sin^2(\al y_1),\\
%&  N_1:=N_x=2\al\bt^2\cos(\al y_1)\sin(\al y_1) + 2\al^2\bt\cosh(\bt y_2)\sinh(\bt y_2),\\
%&  N_2:=N_{xx}=2\al^2\bt^2(\cos^2(\al y_1)-\sin^2(\al y_1) + \cosh^2(\bt y_2) + \sinh^2(\bt y_2)),\\
&  N_3:=N_{xxx}=-8\al^3\bt^2\cos(\al y_1)\sin(\al y_1) + 8\al^2\bt^3\cosh(\bt y_2)\sinh(\bt y_2),\\
&  N_4:=N_{4x}=8\al^2\bt^2(-\al^2\cos^2(\al y_1) + \al^2\sin^2(\al
y_1) + \bt^2\cosh^2(\bt y_2) + \bt^2\sinh^2(\bt y_2)),
\end{align}
\medskip
and
\begin{align}\label{notacionHN7}
%&  H[x,t]:=H=2\Big(\bt\al^2\cosh(\bt y_2)\cos(\al y_1) - \bt^2\al\sinh(\bt y_2)\sin(\al y_1)\Big),\\
%&  H_1:=H_x=-2\al\bt(\al^2+\bt^2)\cosh(\bt y_2)\sin(\al y_1),\\
%&  H_2:=H_{xx}=2\al\bt(\al((\bt^2-\al^2)\cosh(\bt y_2)\cos(\al y_1) -2\al\bt \sin(\al y_1)\sinh(\bt y_2))\nonumber\\
%&- \bt(2\al\bt\cosh(\bt y_2)\cos(\al y_1) +(\bt^2-\al^2)\sin(\al y_1)\sinh(\bt y_2)),\\
&  H_3:=H_{xxx}=2\al\bt((\al^4-\bt^4)\cosh(\bt y_2)\sin(\al y_1) - 2\al\bt(\al^2 + \bt^2)\cos(\al y_1)\sinh(\bt y_2)),\\
&  H_4:=H_{4x}=2\al\bt((\al^5-2\al^3\bt^2-3\al\bt^4)\cosh(\bt y_2)\cos(\al y_1) \nonu\\
&+ (3\al^4\bt+2\al^2\bt^3-\bt^5)\sin(\al y_1)\sinh(\bt y_2)).
\end{align}
\noindent
First of all, remembering from \eqref{h1} that % we start rewriting the following terms of \eqref{7thODEv0}:

\be\label{h17} B_{xx} +2B^3 = \frac{1}{N^3}\Big(2 H^3 + H_2 N^2 - 2 H_1N N_1 + 2 H N_1^2 - H N N_2\Big), \ee
\noindent
we get
\be\label{h27}
\begin{aligned}
&\partial_x^2(B_{xx} +2B^3) = \frac{1}{N^5}\Big(6 H^2 N (H_2 N - 6 H_1 N_1) + 6 H^3 (4 N_1^2 - N N_2)\nonu\\
&+ N (N (H_4 N^2 - 4 H_3 N N_1 + 12 H_2 N_1^2 - 6 H_2 N N_2) - 4 H_1 (6 N_1^3 - 6 N N_1 N_2 + N^2 N_3))\nonu\\
&+ H (24 N_1^4 - 36 N N_1^2 N_2 +  2 N^2 (6 H_1^2 + 3 N_2^2 + 4 N_1 N_3) - N^3 N_4)\Big).
%\Big(2 H^3 + H_2 N^2 - 2 H1 N N_1 + 2 H N_1^2 - H N N_2\Big),
\end{aligned}
\ee \noindent Hence, we have that

\be\label{h37}
\begin{aligned}
& \partial_{x}^2(B_{xx}+ 2B^3) + (4B^2-2(\bt^2 -\al^2))(B_{xx} +2B^3)=\frac{M_1}{N^5},
\end{aligned}
\ee \noindent with \be\label{M1}
\begin{aligned}
&M_1:=\Big(8 H^5 + 2 H^2 N (5 H_2 N - 22 H_1 N_1)+   2 H^3 (16 N_1^2 - 5 N N_2 + 2(\al^2 - \bt^2)N^2)\nonu\\
&+   H \Big[24 N_1^4 - 36 N N_1^2 N_2
+  2 N^2 (6 H_1^2 + 3 N_2^2 + 4 N_1 N_3 +  2(\al^2 - \bt^2) N_1^2 ) - N^3 (N_4 + 2  (\al^2 - \bt^2)N_2)\Big]\nonu\\
& +   N [-24 H_1 N_1^3 + 12 N N_1 (H_2 N_1 + 2 H_1 N_2) +  N^3 (H_4 + 2(\al^2 - \bt^2)H_2) \nonu\\
& - 2 N^2 (2 H_3 N_1 + 3 H_2 N_2 +  2 H_1 (N_3 + (\al^2 -
\bt^2)N_1))]\Big).
\end{aligned}
\ee
\noindent
Moreover, we have that
\be\label{h47}
\begin{aligned}
& - 2B[B_x^2 + B^4]=\frac{-2H}{N^5}\Big(H^4 + (H_1 N - HN_1)^2\Big),
\end{aligned}
\ee
\noindent
and therefore,
\be\label{h57}
\begin{aligned}
& - 2B[B_x^2 + B^4]+(\al^2 +\bt^2)^2B=\frac{M_2}{N^5},
\end{aligned}
\ee
\noindent
with
\be\label{M2}
\begin{aligned}
&M_2:=\Big(H (-2 (H^4 + (H_1 N - H N_1)^2) + (\al^2 + \bt^2)^2N^4)\Big).
\end{aligned}
\ee

Hence, we get the following simplication of \eqref{7thODEv0}:

\begin{align}\label{7thODEv1}
G[B] &= \partial_{x}^2(B_{xx}+ 2B^3) + (4B^2-2(\beta^2-\alpha^2))(B_{xx} +2B^3) - 2B[B_x^2 + B^4]+(\alpha^2+\beta^2)^2B\nonu\\
& = \frac{M_1+ M_2}{N^5},
\end{align}
\noindent with $M_1,~M_2$ in \eqref{M1} and \eqref{M2} respectively. In fact, we verify, using the symbolic software \emph{Mathematica},
that after substituting $H's$ and $N's$ terms explicitly in \eqref{7thODEv1} and lengthy rearrangements, we get

\be\begin{aligned}\label{sumaM1M2v0}
%\begin{aligned}
&~M_1 + M_2=\sum_{i=1}^3p_{ij}\sin(\al y_1)^{2i} + \sum_{i=1}^4 q_{ij}\sin(\al y_1)^{2i-1},\\
&~p_{ij}=\sum_{j=0}^{L_i}a_{ij}\cos(\al y_1)\cosh(\bt y_2)^{2j+1},\qquad q_{ij}=\sum_{j=0}^{L^{'}_i}b_{ij}\sinh(\bt y_2)\cosh(\bt y_2)^{2j},
\quad ~L_i,~L^{'}_i\in\N.
\end{aligned}\ee
\medskip
\noindent
It is easy to see that $a_{ij}=b_{ij}=0,~\forall i=1,\dots,4,~j=0,\dots,L_i,~L^{'}_i$. Therefore we get that

\be\label{sumaM1M2}
%\begin{aligned}
M_1 + M_2=0,
%\end{aligned}
\ee \noindent and we conclude.
\end{proof}

A  direct consequence from Theorem \ref{GBnvbc} and identity \eqref{52ndmkdvnew}, implies that for the 7th and 9th order cases,
we are able to obtain a new identity relating $\tilde{B}_t$ and lower order spatial derivatives of the 7th and 9th-mKdV breathers
(see \eqref{52ndmkdvnew} for comparison):

\begin{cor}
%The following identities associated to 7th and 9th-mKdV breathers
Let $B= B_{\al,\bt}$ be any  7th or 9th-mKdV breather solutions \eqref{579breather} as it corresponds. Then, for any  fixed $t\in \R$, the associated profile $\tilde B$ \eqref{5tBmkdv} to any
 7th or 9th-mKdVbreather satisfies the following nonlinear identities:
\begin{enumerate}
 \item\emph{7th order case:}\\
 %Let $B= B_{\al,\bt}$ be any 7th-mKdV breather solution \eqref{579breather}.
%Then, for any  fixed $t\in \R$, $\tilde B$ satisfies the following nonlinear identity
\be\label{2ndmkdv7}
\begin{aligned}
& \tilde B_{t} -2(\bt^2-\al^2)(\al^2+\bt^2)^2B + 4(\al^4-6\al^2\bt^2+\bt^4)B^3 + 4(\bt^2-\al^2)B^5 - 4B^7\\
&+(3\al^4-10\al^2\bt^2+3\bt^4)B_{xx}  + 4(\bt^2-\al^2)BB_x^2 - 20B^3B_x^2 + 2BB_{xx}^2 - 4BB_xB_{3x} = 0.
\end{aligned}
\ee

\item\emph{9th order case:}\\
 %Let $B= B_{\al,\bt}$ be any 9th-mKdV breather solution \eqref{579breather}.
%Then, for any  fixed $t\in \R$, $\tilde B$ satisfies the following nonlinear identity (REHACER EL CASO 9TH-MKDV)
\be\label{2ndmkdv9}
\begin{aligned}
& ~\tilde B_{t} + a_0 B +a_1 B^3
+a_2 B^5 + 16 \left(\beta ^2-\alpha ^2\right) B^7 -26 B^9 + a_3B_x^2 B +32 \left(\alpha ^2-\beta ^2\right) B_x^2 B^3 - 100 B_x^2 B^5\\
&~  -2 B_x^4 B + a_4 B_{xx} - 6 \left(\alpha ^2+\beta ^2\right)^2 B_{xx} B^2 + 20 \left(\beta ^2-\alpha ^2\right) B_{xx} B^4 - 28 B_{xx} B^6 
+4 \left(\beta ^2-\alpha ^2\right) B_x^2 B_{xx} -12 B_x^2 B_{xx} B^2\\
&~   + 8 \left(\beta ^2-\alpha ^2\right) B_{xx}^2 B - 4 B_{xx}^2 B^3 + 2 B_{xx}^3  + 8 \left(\alpha ^2-\beta ^2\right) B_x B_{3x} B
  - 32 B_x B_{3x} B^3 -4 B_x B_{xx} B_{3x} -2 B_{3x}^2 B=0,
\end{aligned}
\ee
\medskip\noindent
\emph{for}
\[\begin{aligned}
%&\text{for}\\
&~a_0=-\left(\alpha ^2+\beta ^2\right)^2(3\al^4-10\al^2\bt^2+3\bt^4),~ a_1=-4(\al^2-\bt^2)(\al^4-14\al^2\bt^2+\bt^4),\\
&~a_2=-2 \left(\alpha ^4+18 \alpha ^2 \beta ^2+\beta ^4\right), ~a_3=2 \left(5 \alpha ^4-6 \alpha ^2 \beta ^2+5 \beta ^4\right),~ a_4=-4 (\al^2-\bt^2)(\al^4-6\al^2\bt^2+\bt^4).\\
\end{aligned}
\]
\end{enumerate}
\end{cor}

\begin{proof}
For both 7th and 9th order cases, using $B_{4x}$ in \eqref{EcBnvbc}, and computing from it the expressions of $B_{6x}, B_{8x}$, and
substituting recursively $B_{4x}$, we get \eqref{2ndmkdv7} and \eqref{2ndmkdv9}. % Substituting $B_{4x}$ appearing in \eqref{EcBnvbc} into \eqref{EcBnvbc} and simplifying, we get \eqref{2ndmkdv2}.
\end{proof}

\ms
\noindent

\section{Spectral analysis}\label{3}

\medskip

For  any 5th, 7th or 9th-mKdV breather solution $B=B_{\al,\bt}$, we define the  following fourth order linear operator:
\begin{align}\label{L1}
\mathcal L [z](x;t) &  :=  z_{(4x)}(x) -2(\bt^2 -\al^2) z_{xx}(x) +(\al^2 +\bt^2)^2 z(x)  + 10B^2 z_{xx}(x) + 20BB_x z_x(x) \nonu \\
&   \qquad  + \ \big[ 10B_x^2  +20 BB_{xx}  + 30 B^4 -12(\bt^2
-\al^2) B^2 \big] z(x).
\end{align}

\medskip

As a direct consequence of the already studied  spectral properties  of the linearized operator $\mathcal L [z]$, associated to the mKdV breather solution $B$,  in \cite{AM}, we 
obtain the same results for the 5th, 7th or 9th-mKdV breather solutions. In the following lines and for the sake of completeness, we only summarize and list the main features of \eqref{L1}:
 consider first the functions $B_1,~B_2$ \eqref{B12} associated to  5th, 7th and 9th-mKdV breather solutions $B$ (as it corresponds)
% \begin{eqnarray*}
%  B_1(t ; x_1,x_2) := \partial_{x_1} B_{\al,\bt}(t ; x_1, x_2),\quad \hbox{ and } \quad  B_2(t ; x_1,x_2): =\partial_{x_2} B_{\al,\bt}(t ;  x_1, x_2),
% \end{eqnarray*}
%\noindent
and denote as \emph{scaling directions}, the derivatives
\begin{align}\label{DDeltaAB}
\Lambda_\al B =\partial_\al B ,\quad \Lambda_\bt B = \partial_\bt B.
\end{align}

\medskip
\noindent
We get the following 

% Let $z=z(x)$ be a function to be specified in the following lines.
% Let $B=B_{\al,\bt}$ be any 5th or 7th-mKdV breather solution, with
% shift parameters $x_1,x_2$. Let us introduce the following fourth
% order linear operator:
% \begin{align}\label{L1}
% \mathcal L [z](x;t) &  :=  z_{(4x)}(x) -2(\bt^2 -\al^2) z_{xx}(x) +(\al^2 +\bt^2)^2 z(x)  + 10B^2 z_{xx}(x) + 20BB_x z_x(x) \nonu \\
% &   \qquad  + \ \big[ 10B_x^2  +20 BB_{xx}  + 30 B^4 -12(\bt^2
% -\al^2) B^2 \big] z(x).
% \end{align}
% In this section we describe the spectrum of this operator associated to 5th and 7th-mKdV breathers.
% More precisely, our main purpose is to find a suitable coercivity property,
% independently of the nature of scaling parameters. The main result of this section is contained in Proposition \ref{PropOrtog}.
% Part of the analysis carried out in this section has been previously introduced by Lax \cite{LAX1}, and Maddocks and Sachs \cite{MS},
% so we follow their arguments adapted to the breather case, sketching several proofs.

\medskip

\begin{lem}\label{Lspectral} For any  5th, 7th or 9th-mKdV breather solution $B=B_{\al,\bt}$, we get that
\begin{enumerate}
 \item (\emph{Continuous spectrum}) $\mathcal L$ is a linear, unbounded operator in $L^2(\R)$, with dense domain $H^4(\R)$. Moreover, $\mathcal L$ is self-adjoint,
and is a compact perturbation of the constant coefficients operator
\[
\mathcal L_{0} [z]:= z_{(4x)} -2(\bt^2 -\al^2) z_{xx} +(\al^2 +\bt^2)^2 z.
\]
In particular, the continuous spectrum of $\mathcal L$ is the closed interval $[(\al^2 +\bt^2)^2,+\infty)$
in the case $\beta\geq \al$, and $[ 4\al^2 \bt^2 ,+\infty)$ in the case $\beta< \al$, with no embedded eigenvalues are contained in this region.

\medskip

\item (\emph{Kernel}) For each $t\in \R$, one has
\[
\ker \mathcal L =\spawn \big\{ B_1(t;x_1,x_2), B_2(t;x_1,x_2)\big\}.
\]

\medskip

\item Consider the scaling directions $\Lambda_\al B$ and $\Lambda_\bt B$ introduced in \eqref{DDeltaAB}. Then
\be\label{Pos}
\int_\R  \Lambda_\al B \, \mathcal L [\Lambda_\al B]  =  16 \al^2\bt  >0,
\ee
and
\be\label{Neg}
\int_\R  \Lambda_\bt B\, \mathcal L [\Lambda_\bt B]  =  -16 \al^2 \bt <0.
\ee

\medskip

\item Let
\be\label{B0}
B_0 :=  \frac{\al\Lambda_{\bt} B + \bt \Lambda_\al B}{8\al\bt (\al^2 +\bt^2)}.
\ee
Then $B_0$ is Schwartz and satisfies $\mathcal L[B_0] = - B$,
\be\label{negB0}
\int_\R B_0 B = \frac{1}{4\bt (\al^2 +\bt^2)} > 0, \quad \hbox{ and }\quad
 \frac 12\int_\R B_0\mathcal L [B_0]  = - \frac{1}{8\bt(\al^2 +\bt^2)}<0.
\ee

\medskip

\item Let $B_1, B_2$ the  kernel elements defined in \eqref{B12} and $W$  the Wronskian matrix of the functions $B_1$ and $B_2$,
\be\label{WM}
W[B_1, B_2] (t;x) := \left[ \begin{array}{cc} B_1 & B_2 \\  (B_1)_x & (B_2)_x  \end{array} \right] (t,x).
\ee
Then

\be\label{Wsimpl} \det W[B_1,B_2](t;x) =-\frac{ 8\al^3\bt^3 (\al^2
+\bt^2)[ \al\sinh(2 \bt y_2) -\bt \sin(2 \al y_1)]}{(\al^2 +\bt^2
+\al^2 \cosh(2\bt y_2) - \bt^2 \cos(2\al y_1))^2}. \ee

\medskip

\item The operator $\mathcal{L}$ defined in (\ref{L1}) (associated with 5th, 7th and 9th mKdV equations) has a
unique negative eigenvalue $-\lambda_0^2<0$, of multiplicity one, and $\lambda_0=\lambda_0(\alpha,\beta,x_1,x_2,t)$.

\medskip

\item (\emph{Coercivity}) Let us consider the quadratic from associated to $\mathcal{L}$ \eqref{L1}:

\begin{eqnarray}\label{Qmu}\mathcal{Q}[z]&:=
&\int_{\mathbb{R}}z\mathcal{L}[z]=
\int_{\mathbb{R}}z_{xx}^2+2(\beta^2-\alpha^2)\int_{\mathbb{R}}z_x^2+(\alpha^2+\beta^2)^2\int_{\mathbb{R}}z^2-10\int_{\mathbb{R}}B^2z_x^2\nonumber\\
&-&10\int_{\mathbb{R}}B_x^2z^2-40\int_{\mathbb{R}}BB_{x}zz_x+30\int_{\mathbb{R}}B^4z^2-12(\beta^2-\alpha^2)\int_{\mathbb{R}}B^2z^2.\end{eqnarray}

There exists a continuous function $\nu_0 =\nu_0(\al,\bt)$, well-defined and positive for all $\al,\bt>0$ and such that, for all $z_0\in H^2(\R)$ satisfying
\be\label{Or1}
\int_\R z_0 B_{-1} =\int_\R z_0 B_1 =\int_\R z_0 B_2 =0,
\ee
then
\be\label{coee}
\mathcal Q[ z_0] \geq \nu_0\| z_0\|_{H^2(\R)}^2.
\ee
\end{enumerate}
\end{lem}

\ms
\noindent
For the proof of this Lemma, we refer the interested reader to \cite[Sect.4]{AM}.

\medskip

\section{Variational characterization of higher order mKdV breathers}\label{4}

\medskip
In this section we define a $H^2$-Lyapunov functional for both 5th, 7th and 9th-mKdV equations (\ref{5mkdv}), \eqref{7mkdv} and (\ref{9mkdv}) and 
associated to any of the higher order breather solutions. This approach is completely similar to the one depicted in \cite{AM} for the classical
mKdV breather solution.
% Consider  $u_0\in H^2(\R)$ and let $u=u(t) \in H^2(\R)$ be the associated local in time solution of the Cauchy problem associated
% to (\ref{5mkdv}), \eqref{7mkdv} or (\ref{9mkdv}), with initial condition $u(0)=u_0$ (cf. \cite{Lin}, \cite{Kwon}, \cite{Gru}).

\medskip

%Using the functional $E_5[u]$ given in \eqref{E5} or $E_7[u]$ in
%\eqref{E7} which are conserved quantities respectively for 5th and
%7th order mKdV equations, we build a new $H^2$-Lyapunov functional
%specifically associated to the breather solution. Let $B =
%B_{\al,\bt}$ be any 5th or 7th-mKdV breather and  $t\in \R$.
%Consider $M[u]$ and $E[u]$ given in \eqref{M1cor} and \eqref{E1}
%respectively. We define
%\begin{eqnarray}\label{LyapunovGE} \mathcal H[u(t)] := E_i[u](t) + 2(\bt^2-\al^2)E[u](t) + (\al^2 +\bt^2)^2 M[u](t),\ \ i=5,7.\end{eqnarray}
%\noindent

% Using the functional $E_5[u]$ given in \eqref{E5} which is a conserved quantity for both 5th, 7th and 9th order mKdV equations,
% we build a new $H^2$-Lyapunov functional specifically associated to the breather solution.  Consider $M[u]$ and $E[u]$ given in \eqref{M1cor} and \eqref{E1} respectively. We define

Let $B = B_{\al,\bt}$ be any
5th, 7th or 9th-mKdV breather solution and  $t\in \R$. Using a linear combination of the functionals $E_5[u]$, $E[u]$  and  $M[u]$ given in 
\eqref{E5}, \eqref{E1} and \eqref{M1cor}, we define

\medskip
\begin{equation}\label{LyapunovGE} \mathcal H[u(t)] := E_5[u](t) + 2(\bt^2-\al^2)E[u](t) + (\al^2 +\bt^2)^2 M[u](t).\end{equation}
\medskip
\noindent

Therefore, $\mathcal H[u]$ is  a real-valued conserved quantity, well-defined for $H^2$-solutions of (\ref{5mkdv}), \eqref{7mkdv} and (\ref{9mkdv}).
%Note additionally that the functionals $\mathcal{H}_{\mu=0}$ and $\mathcal H_\mu$ for the mKdV and Gardner equations are surprisingly the same.
Moreover, one has the following:

\begin{lem}\label{crit} 5th, 7th and 9th-mKdV breathers \eqref{579Bre} are critical points of the Lyapunov functional $\mathcal H$ \eqref{LyapunovGE}.
In fact, for any $z\in H^2(\R)$  with sufficiently small $H^2$-norm, and $B=B_{\al,\bt}$  any 5th, 7th and 9th-mKdV breather solutions, then,
for all $t\in \R$,  one has \be\label{EE} \mathcal{H}[B +z] -\mathcal{H}[B]  = \frac 12\mathcal Q[z] + \mathcal N[z], \ee with
$\mathcal Q$ being the quadratic form defined in \eqref{Qmu}, and $\mathcal N[z]$ satisfying $|\mathcal N[z] | \leq K\|z\|_{H^2(\R)}^3.$
\end{lem}
\begin{proof}
Considering any 5th, 7th or 9th-mKdV breather $B$, we compute

%\begin{eqnarray*}\mathcal{H}[B+z]&=&\frac{1}{2}\displaystyle\int_{\mathbb{R}}[B_{xx}]^2\ +\displaystyle\frac{1}{2}\displaystyle\int_{\mathbb{R}}[z_{xx}]^2\ +\displaystyle\int_{\mathbb{R}}[B_{xx}
%z_{xx}]-5\displaystyle\int_{\mathbb{R}}[B^2B_x^2]\\
%&-&10\displaystyle\int_{\mathbb{R}}[B^2B_xz_x]-5\displaystyle\int_{\mathbb{R}}[B^2z_x^2]-10\displaystyle\int_{\mathbb{R}}[BB_x^2z]-20\displaystyle\int_{\mathbb{R}}[BB_xzz_x]\\
%&-&10\displaystyle\int_{\mathbb{R}}[Bzz_x^2]-5\displaystyle\int_{\mathbb{R}}[B_x^2z^2]-10\displaystyle\int_{\mathbb{R}}[B_xz^2z_x]-5\displaystyle\int_{\mathbb{R}}[z^2z_x^2]\\
%&+&\displaystyle\int_{\mathbb{R}}B^6+6\displaystyle\int_{\mathbb{R}}[B^5z]+15\displaystyle\int_{\mathbb{R}}[B^4z^2]+20\displaystyle\int_{\mathbb{R}}[B^3z^3]\\
%&+&15\displaystyle\int_{\mathbb{R}}[B^2z^4]+6\displaystyle\int_{\mathbb{R}}[Bz^5]+\displaystyle\int_{\mathbb{R}}z^6+(\beta^2-\alpha^2)\displaystyle\int_{\mathbb{R}}[B_x]^2\\
%&+&(\beta^2-\alpha^2)\displaystyle\int_{\mathbb{R}}[z_x]^2+2(\beta^2-\alpha^2)\displaystyle\int_{\mathbb{R}}[B_xz_x]-(\beta^2-\alpha^2)\displaystyle\int_{\mathbb{R}}B^4\\
%&-&4(\beta^2-\alpha^2)\displaystyle\int_{\mathbb{R}}[B^3z]\
%-6(\beta^2-\alpha^2)\displaystyle\int_{\mathbb{R}}[B^2z^2]-4(\beta^2-\alpha^2)\displaystyle\int_{\mathbb{R}}[Bz^3]\\
%&-&(\beta^2-\alpha^2)\displaystyle\int_{\mathbb{R}}z^4+\displaystyle\frac{(\alpha^2+\beta^2)^2}{2}\displaystyle\int_{\mathbb{R}}B^2+\displaystyle\frac{(\alpha^2+\beta^2)^2}{2}\displaystyle\int_{\mathbb{R}}z^2+{(\alpha^2+\beta^2)^2}\displaystyle\int_{\mathbb{R}}[B
%z].\end{eqnarray*}
$$\begin{aligned}
&\mathcal{H}[B+z] =  \frac 12 \int_\R (B+z)_{xx}^2 -5 \int_\R (B+z)^2(B+z)_x^2 +  \int_\R (B+z)^6\\
& + (\bt^2-\al^2)\frac 12  \int_\R (B+z)_x^2 -  (\bt^2 -\al^2)
\int_\R (B+z)^4  + \frac 12 (\al^2 +\bt^2)^2 \int_\R (B+z)^2\\
& = \frac 12 \int_\R B_{xx}^2 -5 \int_\R B^2B_x^2 +  \int_\R B^6
+ (\bt^2-\al^2)\frac 12  \int_\R B_x^2 - \frac 12 (\bt^2 -\al^2) \int_\R B^4  + \frac 12 (\al^2 +\bt^2)^2 \int_\R B^2\\
&  + \int_\R z\Big[ B_{4x} + 10 B B_{x}^2+ 10B^2 B_{xx} + 6 B^5 -   2(\beta^2 -\al^2) (B_{xx} + 2B^3) + (\al^2 + \bt^2)^2 B\Big]\\
&  +\frac 12 \Big[  \int_\R z_{xx}^2 + 2(\bt^2 -\al^2)\int_\R z_{x}^2 +(\al^2 +\bt^2)^2\int_\R z^2 + 10\int_\R B^2 z_{xx}z\\
&   -20\int_\R B B_{x} z_xz+\int_\R( 30 B^4 -10B_{x}^2 - 12(\beta^2 -\al^2)B^2)z^2\Big]\\
&  -\frac52\int_\R( z^2z_x^2 + 2B_{x}z^2z_x + 2B zz_x^2) + \int_\R 5B^3z^3 + \frac{15}{4}B^2z^4 +\frac 32 \int_\R B z^5 +\frac 14 \int_\R z^6\\
&  -2(\bt^2-\al^2)\int_\R B z^3  -\frac 12 (\bt^2-\al^2)\int_\R z^4.
\end{aligned}$$

We finally obtain:
\begin{eqnarray*}\mathcal{H}[B+z]=\mathcal{H}[B]+\int_{\mathbb{R}}G[B]z(t)\
dx+\frac{1}{2}\mathcal{Q}[z]+\mathcal{N}[z],\end{eqnarray*} where
$\mathcal{Q}$ is defined in \eqref{Qmu} and
\begin{eqnarray*}
G[B]:=B_{4x}+10BB_x^2+10B^2B_{xx}+6B^5-2(\beta^2-\alpha^2)(B_{xx}+2B^3)+(\alpha^2+\beta^2)^2B.\end{eqnarray*}

From Theorem (\ref{GBnvbc}), one has $G[B]\equiv 0$. Finally,
the term $\mathcal{N}[z]$ is given by
\begin{eqnarray}
\mathcal{N}[z]:&=&-10\int_{\mathbb{R}}Bzz_x^2+\frac{10}{3}\int_{\mathbb{R}}[B_{xx}z^3-5z^2z_x^2
+20B^3z^3+15B^2z^4+6Bz^5+z^6]\nonu\\
&&-4(\beta^2-\alpha^2)\int_{\mathbb{R}}Bz^3-(\beta^2-\alpha^2)\int_{\mathbb{R}}z^4.
\end{eqnarray}

Therefore, from direct estimates one has
$|\mathcal{N}[z]|\leq\mathcal{O}(\|z\|^3_{H^2(\mathbb{R})})$ as
desired.

\end{proof}

%The previous Lemma is the key step to the proof of the main result of this paper, that we develop in the next section.
Using the previous Lemma, we are able to prove the main result of the paper.
\bigskip%

\section{Main Theorem}\label{5}

\medskip

%In this section we prove a detailed version of  Theorem \ref{T1p8}.

\begin{thm}[$H^2$-stability of 5th, 7th and 9th order mKdV breathers]\label{T1gardner}

Let $\al, \bt \in \R\backslash\{0\}$  and $B=B_{\alpha,\beta}$ any 5th, 7th or 9th order mKdV breather. %$\mu$ such that $\mu\in(0,\mu_{\max})$.%\frac{\al^2+\bt^2}{\sqrt{2(\al^2+4\bt^2)}})$.
There exist positive parameters $\eta_0, A_0$, depending on $\al$ and $\beta$, such that the following holds.  Consider $u_0 \in
H^2(\R)$, and assume that there exists $\eta \in (0,\eta_0)$ such that 

\be\label{In} \|  u_0 - B(t=0;0,0) \|_{H^2(\R)} \leq \eta. \ee

Then there exist $x_1(t), x_2(t)\in \R$ such that the solution $u(t)$ of the Cauchy problem for the 5th  \eqref{5mkdv}, 7th \eqref{7mkdv}
or for the 9th \eqref{9mkdv}  equations, with initial data $w_0 \in H^2(\R)$, satisfies 

\be\label{Fn1} \sup_{t\in \R}\big\| w(t) - B(t;x_1(t),x_2(t)) \big\|_{H^2(\R)}\leq A_0 \eta, \ee with \be\label{Fn2} \sup_{t\in \R}|x_1'(t)| +|x_2'(t)| \leq KA_0 \eta,
\ee for a constant $K>0$.

\end{thm}

\begin{rem}
% The initial condition (\ref{In}) can be replaced by any initial breather profile of the form $\hat{B} := B_{\al,\bt}(t_0; x_1^0,
% x_2^0)$, with $t_0, x_1^0, x_2^0 \in \R$, thanks to the invariance of the equation under translations in time and space. 

Note that the same result is true for the \emph{negative} breather $-B_{\al,\bt}$ which is also a solution of (\ref{5mkdv}) or
(\ref{7mkdv}).
\end{rem}

\medskip

\begin{proof}[Proof of Theorem \ref{T1gardner}] %PENSAR SE DEIXAMOS A PROVA ASSIM
We take $u=u(t) \in H^2(\R)$ as the associated local in time solution of the Cauchy problem associated
 to (\ref{5mkdv}), \eqref{7mkdv} or (\ref{9mkdv}), with  initial condition $u(0)=u_0\in H^2(\R)$ (cf. \cite{Lin}, \cite{Kwon}, \cite{Gru}).
%\end{proof}
Therefore once we guaranteed for the case of 5th, 7th and 9th-mKdV breathers, that they satisfy the same 4th order stationary ODE \eqref{EcBnvbc} as the \emph{classical} mKdV breather,
that a suitable coercivity property holds for the bilinear form $\mathcal Q$ associated to any of these higher order breathers (see \eqref{coee}), and the existence of a
unique negative eigenvalue \eqref{Lspectral} of the linearized operator $\mathcal L$ associated again to these higher order breathers, the stability proof follows the same steps, 
namely, we proceed assuming that the maximal time of stability $T$ is finite and we arrive to a contradiction. In fact it is completely similar as the $H^2$-stability of classical mKdV breathers \cite[Theorem 6.1]{AM}.

\end{proof}

\appendix

\section{11th-mKdV equation}\label{9mkdvApp}
% The 9th order mKdV equation is written as follows
% 
% \begin{align}
% u_t &+ \partial_x\Big(u_{8x} + 18u^2u_{6x} + 108uu_xu_{5x} +228uu_{2x}u_{4x} + 210(u_x)^2u_{4x} + 126u^4u_{4x} + 138u(u_{3x})^2 \\
% & + 756u_xu_{2x}u_{3x} + 1008u^3u_xu_{3x} + 182(u_{2x})^3 + 756u^3(u_{2x})^2 + 3108u^2(u_x)^2u_{2x} + 420u^6u_{2x} \\
% & + 798u(u_x)^4 + 1260u^5(u_x)^2 + 70u^9\Big)=0. \label{9th}
% \end{align}
% 
% \begin{defn}[9th-mKdV breather]\label{9breatherApp} Let $\al, \bt >0$
% and $x_1,x_2\in \R$. The real-valued  breather solution of the 9th-mKdV equation \eqref{9mkdvApp} is given explicitly by the formula
% \be\label{9Bre}B\equiv B_{\al, \bt}(t,x;x_1,x_2)   %:=   \partial_x \tilde B_\mu
%  :=   2\partial_x\Bigg[\arctan\Big(\frac{\bt}{\al}\frac{\sin(\al y_1)}{\cosh(\bt y_2)}\Big)\Bigg],
% \ee
% with $y_1$ and $y_2$
% \be\label{y1y2GE9App}
% \begin{aligned}
% &y_1 = x+ \delta_9 t + x_1, \quad y_2 = x+ \ga_9 t + x_2,~~~\end{aligned}\ee
% \noindent and with velocities
% 
% 
% \be\label{speeds9App}
% \begin{aligned}
% &\delta_9 := -\al^8 + 36\al^6\bt^2 - 126\al^4\bt^4 + 84\al^3\bt^6-9\bt^8,\quad\ga_9 :=-\bt^8+36\al^2\bt^6-126\al^4\bt^4 + 84\al^6\bt^2-9\al^8.\\
% \end{aligned}\ee
% %&\text{and}\quad\\
% % \noindent and  $(\delta_7,\ga_7)$ in the 7th order case
% % \be\label{speeds57-1}
% % \begin{aligned}
% % &\delta_7 := \al^6-21\al^4\bt^2+35\al^2\bt^4-7\bt^6,\quad\ga_7 :=-\bt^6+21\al^2\bt^4-35\al^4\bt^2+7\al^6.
% %&\delta := -\al^4+10\al^2\bt^2-5\bt^4,\quad\text{and}\quad\ga :=-\bt^4+10\al^2\bt^2-5\al^4.
% \end{defn}
% 
% \bigskip
For the sake of completeness, we show the 11th order mKdV equation. It is written as follows:

\medskip

\be\label{y1y2GE11}
\begin{aligned}
&u_t + \partial_x\Big(
u_{10x}+22 u^2 u_{8x}+198 u^4u_{6x}+924 u^6 u_{4x}+506 u \left(u_{4x}\right)^2
+3036 u^3 \left(u_{3x}\right)^2+2310 u^8 u_{xx}\\
&+8316 u^5 \left(u_{xx}\right)^2+9372 u^2 \left(u_{xx}\right)^3+9240 u^7 \left(u_x\right)^2
+26796 u^3 \left(u_x\right)^4+176 u u_x u_{7x}+484 u 
u_{xx} u_{6x}+462 \left(u_x\right)^2 u_{6x}\\
&+836 u u_{3x}u_{5x}+2376 u^3 u_xu_{5x}
+5016 u^3 u_{xx} u_{4x}
+2706 \left(u_{xx}\right)^2 u_{4x}+11220 u^2 \left(u_x\right)^2 u_{4x}+3498 u_{xx} \left(u_{3x}
\right)^2\\
&+11088 u^5 u_x u_{3x}+21120 u \left(u_x\right)^3 u_{3x}+54516 u^4 \left(u_x\right)^2 u_{xx}
+44748 u \left(u_x\right)^2 \left(u_{xx}\right)^2+13398 \left(u_x\right)^4 u_{xx}\\
&+2376 u_x u_{xx} u_{5x}+3696 u_x u_{3x} u_{4x}+39336 u^2 u_x u_{xx} u_{3x}+252 u^{11}\Big)=0.
\end{aligned}\ee

\medskip
\noindent
Moreover, we are able to obtain the 11th order mKdV breather solution, in the same way we used to get \eqref{579Bre}:

\begin{defn}[11th-mKdV breather]\label{11breather} Let $\al, \bt >0$
and $x_1,x_2\in \R$. The real-valued  breather solution of the 11th-mKdV equation \eqref{9mkdvApp} is given explicitly by the formula
\be\label{11Bre}B\equiv B_{\al, \bt}(t,x;x_1,x_2)   %:=   \partial_x \tilde B_\mu
 :=   2\partial_x\Bigg[\arctan\Big(\frac{\bt}{\al}\frac{\sin(\al y_1)}{\cosh(\bt y_2)}\Big)\Bigg],
\ee
% \[ 
%  u(t,x):=2\partial_x\Bigg[\arctan\Big(\frac{\bt}{\al}\frac{\sin(\al y_1)}{\cosh(\bt y_2)}\Big)\Bigg],
% \]
with $y_1$ and $y_2$
\be\label{y1y2GE11}
\begin{aligned}
&y_1 = x+ \delta_{11} t + x_1, \quad y_2 = x+ \ga_{11} t + x_2,\end{aligned}\ee
\noindent and with velocities
\be\label{speeds11}
\begin{aligned}
& \delta_{11} = \alpha ^{10}-55 \alpha ^8 \beta ^2+330 \alpha ^6 \beta ^4-462 \alpha ^4 \beta ^6 +165 \alpha ^2 \beta ^8-11\bt^{10},\\
 & \gamma_{11} = 11 \alpha ^{10} -165 \alpha ^8 \beta ^2+462 \alpha ^6 \beta ^4-330 \alpha ^4 \beta ^6+55 \alpha ^2 \beta ^8-\beta ^{10}.
%  & \delta_{11} = -\alpha ^{11}+55 \alpha ^9 \beta ^2-330 \alpha ^7 \beta ^4+462 \alpha ^5 \beta ^6 -165 \alpha ^3 \beta ^8+11 \alpha  \beta ^{10},\\
%  & \gamma_{11} = -11 \alpha ^{10} \beta +165 \alpha ^8 \beta ^3-462 \alpha ^6 \beta ^5+330 \alpha ^4 \beta ^7-55 \alpha ^2 \beta ^9+\beta ^{11}.
\end{aligned}\ee
\end{defn}

\end{document}